\documentclass[10pt]{article}

\usepackage{tikz}
\usepackage{color}
\usepackage{amsmath,amsfonts,amsthm,amsbsy,amssymb,dsfont}
\usepackage[colorlinks=true,urlcolor=webblue,linkcolor=webgreen,filecolor=webblue,citecolor=webgreen,pdfpagemode=UseOutlines,pdfstartview=FitH,pdfpagelayout=OneColumn,bookmarks=true]{hyperref}
\usepackage{fullpage}
\usepackage{doi}
\usepackage[numbers]{natbib}
\usepackage{microtype}
\usepackage{xspace}
\usepackage{caption,subcaption}
\usepackage{lmodern}
\usepackage{authblk}

\colorlet{darkblue}{blue!70!black}
\colorlet{darkred}{red!70!black}

\usetikzlibrary{arrows,decorations.pathmorphing}

\tikzstyle{tangle}=[draw=white,double=darkblue,line width=.08cm,double distance=1.4pt]
\tikzstyle{tanglered}=[draw=white,double=darkred,line width=.08cm,double distance=1.4pt]
\tikzstyle{stangle}=[style=tangle,scale=0.75]
\tikzstyle{tanglea}=[style=tangle,>=stealth,->]
   
\tikzstyle{basic}=[draw=darkblue,fill=darkblue,semithick,scale=.75,>=stealth]

\hypersetup{
  pdftitle=Quantum Fourier Transforms and the Complexity of Link Invariants for Quantum Doubles of Finite Groups,
  pdfauthor=Hari Krovi and Alexander Russell}

\definecolor{webgreen}{rgb}{0,.5,0}
\definecolor{webblue}{rgb}{0,0,.5}

\DeclareMathOperator{\tr}{Tr}

\numberwithin{equation}{section}

\newtheorem{theorem}{Theorem}
\newtheorem{lemma}[theorem]{Lemma}

\newtheorem{definition}{Definition}
\newcommand{\one}{\mathds 1}
\renewcommand{\vec}[1]{#1}
\newcommand{\D}{\mathsf{D}}
\newcommand{\Ind}{\textrm{Ind}}
\newcommand{\polylog}{\operatorname{polylog}}
\newcommand{\wreath}{\operatorname{wr}}
\newcommand{\remove}[1]{}
\newcommand{\Pl}{\textsl{Pl}}

\newcommand{\C}{\mathbb{C}}
\newcommand{\R}{\mathbb{R}}
\newcommand{\Z}{\mathbb{Z}}
\newcommand{\U}{\textsf{U}}

\newcommand{\ket}[1]{|#1\rangle}
\newcommand{\bra}[1]{\langle #1|}
\newcommand{\Hom}{\operatorname{Hom}}
\newcommand{\GL}{\textsf{GL}}
\newcommand{\End}{\textsf{End}}
\newcommand{\SU}{\textsf{SU}}

\newcommand{\BQP}{\mathsf{BQP}}
\newcommand{\BPP}{\mathsf{BPP}}
\newcommand{\sP}{\mathsf{\#P}}
\newcommand{\SBQP}{\mathsf{SBQP}}
\newcommand{\SBP}{\mathsf{SBP}}

\newcommand{\st}{\textsuperscript{1}}
\newcommand{\dg}{\textsuperscript{2}}

\newcommand{\revise}[1]{}

\begin{document}

\title{Quantum Fourier Transforms and the Complexity of \\Link Invariants for Quantum Doubles of Finite Groups}
\author{Hari Krovi\st \thanks{Email: hkrovi@bbn.com} } 
\author{Alexander Russell\dg \thanks{Email: acr@cse.uconn.edu}}
\affil[]{\st Raytheon BBN Technologies \\
	Cambridge, MA\\
	\vspace{0.1in}
	\dg University of Connecticut\\
	Storrs, CT}
\maketitle

\begin{abstract}
Knot and link invariants naturally arise from any braided Hopf algebra. We consider the computational complexity of the invariants arising from an elementary family of finite-dimensional Hopf algebras: \emph{quantum doubles of finite groups} (denoted $\D(G)$, for a group $G$). These induce a rich family of knot invariants and, additionally, are directly related to topological quantum computation.

Regarding algorithms for these invariants, we develop quantum circuits for the quantum Fourier transform over $\D(G)$; in general, we show that when one can uniformly and efficiently carry out the quantum Fourier transform over the centralizers $Z(g)$ of the elements of $G$, one can efficiently carry out the quantum Fourier transform over $\D(G)$. We apply these results to the symmetric groups to yield efficient circuits for the quantum Fourier transform over $\D(S_n)$. With such a Fourier transform, it is straightforward to obtain additive approximation algorithms for the related link invariant. 

As for hardness results, first we note that in contrast to those such as the Jones polynomial, where the images of the braid group representations are dense in the unitary group, the images of the representations arising from $\D(G)$ are finite. This important difference appears to be directly reflected in the complexity of these invariants. While additively approximating ``dense'' invariants is $\BQP$-complete and multiplicatively approximating them is $\sP$-complete, we show that certain $\D(G)$ invariants (such as $\D(A_n)$ invariants) are $\BPP$-hard to additively approximate, $\SBP$-hard to multiplicatively approximate, and $\sP$-hard to exactly evaluate. To show this, we prove that, for groups (such as $A_n$) which satisfy certain properties, the probability of success of any randomized computation can be approximated to within any $\epsilon$ by the plat closure.

Finally, we make partial progress on the question of simulating anyonic computation in groups uniformly as a function of the group size. In this direction, we provide efficient quantum circuits for the Clebsch-Gordan transform over $\D(G)$ for ``fluxon" irreps, i.e., irreps of $\D(G)$ characterized by a conjugacy class of $G$. For general irreps, i.e., those which are associated with a conjugacy class of $G$ \emph{and} an irrep of a centralizer, we present an efficient implementation under certain conditions such as when there is an efficient Clebsch-Gordan transform over the centralizers (this could be a hard problem for some groups). We remark that this also provides a simulation of certain anyonic models of quantum computation, even in circumstances where the group may have size exponential in the size of the circuit.
\end{abstract}

\section{Introduction}
Quantum computing using anyons was introduced by Kitaev in his seminal paper~\cite{Kitaev:Anyon}. He showed how to construct a \emph{fault tolerant} computer based on anyons where unitary transformations are performed by braiding the anyons with each other. The anyons considered in the paper arise from quantum doubles of  finite groups. Due to the topological nature of braiding, this model is called topological quantum computing. \citet{OgburnPreskill} showed how to perform \emph{universal} quantum computation in this model using simple groups. In particular, they showed how to simulate a gate set using the alternating group $A_5$. \citet{Mochon2} (see also~\cite{Mochon1}) showed that groups which are solvable but not nilpotent can also be used for universal quantum computation.

In a series of papers, Freedman, Kitaev, Larsen and Wang laid the mathematical foundations of topological quantum computing. In \cite{FKLW} and \cite{Freedman:SimulationTQFT}, it was shown that quantum computers based on the circuit model can simulate topological quantum field theories (TQFT) and hence that topological quantum computers can be simulated using the standard circuit model. This, implicitly, gives an algorithm to approximate the Jones polynomial of knots and links. Aharonov, Jones and Landau~\cite{AJL} later gave an explicit combinatorial algorithm to additively approximate the Jones polynomial. Wocjan and Yard~\cite{WY} also present an efficient algorithm which is based on irreducible representations of Hecke algebras (which give rise to representations of the braid group). In the other direction, Freedman, Kitaev and Wang~\cite{Freedman:modular} show that topological quantum computers based on certain TQFTs can simulate conventional quantum computers. This implies that approximating (additively) the Jones polynomial and plat closures of braids is $\BQP$-complete. In \cite{Freedman:modular}, $\BQP$-completeness was shown when the TQFT was defined at a fixed root of unity $k$.  However, the algorithm in \cite{AJL} works for an asymptotically growing $k$ as well. This gap was closed when \citet{AharonovArad} showed that even when $k$ grows asymptotically, the complexity of approximating the link invariants is $\BQP$-complete. All the hardness results are based on one central fact. The representations of the braid group in all these models have the property that their image is dense in the unitary group. In \cite{Kuperberg:density}, Kuperberg showed that this is the case in general: whenever the image of the link invariant is dense, additive approximation is $\BQP$-complete and multiplicative approximation is $\SBQP$-complete. Using Aaronson's result \cite{Aaronson:PostBQP}, multiplicative approximations are also $\sP$-hard.

However, very little is known about the complexity of non-dense invariants such as those arising from quantum doubles of finite groups. In this paper, we prove some analogous results for computation with anyons arising from quantum doubles of finite groups (denoted $\D(G)$). We provide algorithms and hardness results for approximating link invariants in this model. The main difference between the model in this paper and those referred to above is that the image of the representations of the braid group in this model is not dense (in fact, it is a finite group \cite{FiniteImage}). This means that none of the techniques developed for dense invariants automatically carry over to this model and we develop them in this paper. On the side of algorithms, we first give an efficient circuit for the Fourier transform over the regular representation of $\D(G)$. We then present quantum algorithms to approximate link invariants arising from irreps of $\D(G)$. For certain kinds of irreps, namely those associated with a conjugacy class of $G$ (``fluxon" irreps), we also give classical randomized algorithms to approximate link invariants. For general irreps, there may not exist efficient classical algorithms.

Then turning to hardness results, we show that additive approximations to link invariants for certain groups $G$ which satisfy certain properties are $\BPP$-hard and that multiplicative approximations are $\SBP$-hard and exact evaluations are $\sP$-hard. All the algorithms presented in the paper are efficient in $\log(|G|)$ as well as the number of crossings of the link, where $|G|$ is the size of the group. However, our hardness results require that the group be of constant size.

Finally, we address the question of simulating anyonic quantum computation. Here, we assume that $G$ is in a sequence of groups of asymptotically growing size (since if $G$ is a fixed group, simulation has already been shown \cite{PreskillNotes}). In order to perform an efficient simulation, one needs to perform an efficient Clebsch-Gordan transform. Here, we make partial progress on this issue. In particular, we give the Clebsch-Gordan transform for fluxon irreps of $\D(G)$. For general irreps i.e., those characterized by a conjugacy class of $G$ and an irrep of the centralizer, we give an efficient Clebsch-Gordan transform under certain conditions. If we can
\begin{enumerate}
\item perform efficient QFT and the Clebsch-Gordan transform over every $Z(g)$ and $Z(g)\cap Z(h)$ and,
\item block diagonalize irreps of centralizers restricted to intersections of centralizers. 
\end{enumerate}
To explain the last condition, note that if $\rho$ is an irrep of $Z(g)$, then when restricted to $Z(g)\cap Z(h)$ (a subgroup of $Z(g)$), it breaks up into irreps of the latter group. The last condition, then, says that we must be able to perform the transform that block diagonalizes $\rho$ into blocks of irreps of $Z(g)\cap Z(h)$. This (rather technical) condition can probably be removed. However, in certain groups it can be quite challenging to satisfy the above two conditions. Indeed, as pointed to us by an anonymous referee, it may be quite hard to even find the intersections $Z(g)\cap Z(h)$ in some groups. If the Fourier and Clebsch-Gordan transforms can be performed efficiently, then one can simulate anyonic quantum computation efficiently inside even exponentially large irreps of $\D(G)$ in the circuit model. 

This paper is organized as follows. In Section~\ref{QuantumDoubles}, we introduce the quantum doubles of finite groups in terms of Hopf algebras. Following this, in Section~\ref{Fourier}, we derive their irreducible representations and describe the action of the associated $R$-matrices, which are relevant for braiding. We then describe a quantum circuit for the Fourier transform over $\D(G)$, and show how its complexity relates to the quantum Fourier transform over subgroups of $G$. In Section~\ref{algorithms}, we use this Fourier transform to give quantum algorithms to additively approximate link invariants arising from irreps of $\D(G)$. We also give a classical randomized algorithm for fluxon irreps. Then in Section~\ref{complexity}, we show that additive approximations of these link invariants are $\BPP$-hard, multiplicative approximations are $\SBP$-hard, and exact evaluations are $\sP$-hard. In Section~\ref{QC:anyons} we address the question of simulation of anyonic quantum computation. In Section~\ref{sec:CG}, we describe quantum algorithms for the Clebsch-Gordan transform over $D(G)$. Finally, in Section~\ref{conclusions} we present our conclusions and some open problems. 

\paragraph{Related work on finite image representations of the braid group.}
Finite image representations have been considered in recent papers. In \cite{Rowell}, there is a discussion of two paradigms, one involving dense images of the braid group and the other involving finite images, along with conjectures on the complexity of link invariants. In \cite{RW}, Rowell and Wang show how to ``localize'' certain finite image representations of the braid group. In \cite{HNW2} (see also \cite{HNW1}), Hastings, Nayak and Wang show that link invariants coming from certain finite image representations of the braid group discussed in \cite{RW} are $\sP$ hard to evaluate exactly (and to evaluate a sufficiently good multiplicative approximation). In \cite{FRW}, the image of the Ising anyon representation has been studied.

\section{Quantum doubles of finite groups}\label{QuantumDoubles}
\subsection{Quantum doubles as Hopf algebras}\label{Hopf}

We recall the notion of Hopf algebra; more complete descriptions can be found in accounts by \citet{Kassel} and \citet{Majid}. A $\C$-vector space $A$ is a \emph{Hopf algebra} if it possesses consistent algebra and coalgebra structure augmented with an \emph{antipode} map that yields a natural notion of ``inversion.'' To be more precise, a Hopf algebra possesses the following structure: 
\begin{description}
\item[Algebra structure] A multiplication map $\mu:A \otimes A \rightarrow A$ and a unit map $\eta:\C \rightarrow A$ which satisfy
  \begin{description}
    \item[(Associativity)] $\mu(\mu\otimes\one)(a\otimes b\otimes c)=\mu(\one \otimes\mu)(a\otimes b\otimes c)$ for all $a,b,c\in A$, and
    \item[(Unit)] $\mu(a\otimes \one)=\mu(\one\otimes a)=a$, for all $a \in A$, where $\one$ is the unique element of $A$ for which $\eta(c) = c\one$ for $c \in \C$.
    \end{description} 
\item[Coalgebra structure] A comultiplication rule $\Delta:A\rightarrow A\otimes A$ and a counit $\epsilon:A\rightarrow \C$ which satisfy
\begin{description}
\item[(Coassociativity)] $(\one\otimes\Delta)\Delta=(\Delta\otimes\one)\Delta$, and 
\item[(Counit)] $(\one\otimes\epsilon)\Delta= (\epsilon\otimes \one)\Delta = \one$.
\end{description}
\item[Coherence and an antipode]  The algebra and coalgebra structure are related by the following two axioms:
\begin{description}
\item[(Coherence)] The comultiplication map $\Delta$ and the counit $\epsilon$ are algebra homomorphisms (where $A \otimes A$ is given the natural tensor product algebra structure) and $\eta(\epsilon(a)) = \epsilon(a) \one$.
\item[(Antipode)] An \emph{antipode} map $S: A \rightarrow A$ satisfying
 \[
 \mu(S\otimes\one)\Delta(a)=\mu(\one\otimes S)\Delta(a)=\epsilon(a)\one \,.
\]
\end{description}
\end{description}
We will denote $\mu(a\otimes b)$ as $ab$. It follows from the axioms above that the antipode map $S: H \rightarrow H$ is an antihomomorphism:
\begin{equation}
  \label{eq:S-antihomomorphism}
  S(a)S(b) = S(ba)\,.
\end{equation}
Furthermore, when $H$ is finite-dimensional (the only case we shall consider), $S$ is invertible as a linear operator.

One perspective on the role played by the coalgebra and antipode structure afforded by a Hopf algebra is that it provides ring structure to the family of representations of the underlying algebra. Specifically, the comultiplication operator provides the structure of a representation to the tensor product of two representations $X$ and $Y$ of $A$ by defining the action of $a \in A$ to be that of $\Delta(a)$. The coherence axiom $\Delta(a)\Delta(b) = \Delta(ab)$ guarantees that this yields representation structure; the coassociativity axiom yields a canonical isomorphism between the representation $(X \otimes Y) \otimes Z$ and $X \otimes (Y \otimes Z)$. With this notion of tensor product of representations, the counit axiom guarantees that the map $\epsilon: A \rightarrow \C$, a one-dimensional representation of $A$, is a unit under tensor product: $X \otimes \epsilon \cong X = \epsilon \otimes X$. Finally, the antipode map gives the structure of a representation to $\Hom(X, \C)$ for any representation $X$: the action of $a \in A$  on an element $f \in \Hom(X, \C)$ is given by $(af)(x) =  f(S(a)x)$. The fact that this definition defines a representation depends on~\eqref{eq:S-antihomomorphism}.

A finite group $G$ is naturally associated with two, generally distinct, Hopf algebras. The first is the group algebra, the algebra of formal $\C$-linear combinations of group elements denoted $\C[G] = \{ \sum_g \alpha_g g \mid g \in \C\}$. The second is the dual of this algebra, the set of maps from $G$ to $\C$. A natural basis for this space is given by the delta functions $g^*$, where $g^*(h)=\delta_{g,h}$. We denote this algebra $\C G$ (without brackets). 

Algebra structure on $\C[G]$ is given by linearly extending the group multiplication rule; Hopf algebra structure on $\C[G]$ is defined by adopting the maps
\[
\Delta(g)=g\otimes g, \quad \epsilon(g)=1 \quad \text{and} \quad S(g)=g^{-1} \,.
\]
Observe that these choices correspond to the familiar notions of tensor product of representations, the trivial representation, and representation of dual spaces. It is easy to check that when $G$ is non-abelian, this Hopf algebra is non-commutative. Its comultiplication structure, however, is \emph{cocommutative}, which is to say that
\begin{equation}
T^{-1}\Delta(g)T=\Delta(g)\,,
\end{equation}
where $T$ is the linear operator that swaps the two tensor copies $A\otimes A$.

Now consider the dual algebra $\C G$; the algebra structure is given by the rules
\[
g_1^*g_2^* = \delta_{g_1,g_2} \quad \text{and} \quad \one =  \sum_g g^*\,.
\]
The coalgebra structure is given by the maps
\[
\Delta(g^*)=\sum_{g_1g_2=g}g_2^*\otimes g_1^*, \quad \epsilon(g^*)=\delta_{g,e} \quad\text{and}\quad S(g^*)=(g^{-1})^* ,
\]
where $e$ is the identity element of $G$. For any nonabelian finite group $G$, $\C G$ is a commutative, non-cocommutative Hopf algebra.

With any finite-dimensional Hopf algebra $H$ one may associate a natural dual Hopf algebra $H^*$: each map $\phi: V \rightarrow W$ in the definition of $H$ yields a dual map $\phi^*: \Hom(W, \C) \rightarrow \Hom(V, \C)$. One can check that this process carries the algebra structure of $H$ to the coalgebra structure of $H^*$, the coalgebra structure of $H$ to the algebra structure of $H^*$, and that the resulting maps satisfy the axioms. The two algebras described above are duals of one another.

The \emph{quantum double} $\D(G)$ is an algebra defined on the vector space $\C[G] \otimes \C G$; we write the element $g \otimes h^*$ simply as $gh^*$, letting the superscript on the $h$ remind us that this is an element of the dual algebra. The Hopf algebra structure on $\D(G)$ is obtained by stitching together the $\C[G]$ and $\C G$ structures as follows:
\[
(g_1h_1^*)(g_2h_2^*) = g_1g_2(h_1^{g_2})^*h_2^* = \delta_{h_1^{g_2},h_2}g_1g_2h_2^* \quad \text{and} \quad \one =  e\sum_h h^* 
\]
where $x^y = y^{-1}xy$ denotes the conjugation of $x$ by $y$ and
\begin{equation}\label{eq:double-definition}
\Delta(gh^*)=\sum_{h_1h_2=h}gh_2^*\otimes gh_1^*, \quad \epsilon(gh^*)=\delta_{h,e} \quad\text{and}\quad S(gh^*)=(g^{-1})(gh^{-1}g^{-1})^* \,.
\end{equation}

\newcommand{\cop}{\operatorname{cop}}
\paragraph{Remarks.} If $H$ is a Hopf algebra with comultiplication rule $\Delta$, one obtains the \emph{coopposite} of this algebra, denoted $H^{\cop}$, by \emph{reversing} the comultiplication rule: specifically, the comultiplication structure of $H^{\cop}$ is given by $T \circ \Delta$, where $T: H \otimes H \rightarrow H \otimes H$ is the exchange operator (that linearly extends the rule $T: \alpha \otimes \beta \mapsto \beta \otimes \alpha$). All other algebraic structure is inherited from $H$. In fact, it is this reversed comultiplication structure of $\C G^{\cop}$ featured in the definition~\eqref{eq:double-definition} of the quantum double construction. $\D(G)$ is generated, as an algebra, by the elements $\{ \sum_h gh^* \mid g \in G\}$, an embedded copy of $\C[G]$, and $\{ 1h^* \mid h \in G\}$, an embedded copy of $\C G^{\cop}$.

If $G$ is non-commutative, then the quantum double $\D(G)$ is a non-commutative, non-cocommutative Hopf algebra. The construction is a rough analogue of the semidirect product construction of two groups: specifically, observe that one can interpret the multiplication rule
$$
(g_1h_1^*)(g_2h_2^*) = g_1 g_2 (h_1^{g_2})^* h_2^* = \delta_{h_1^{g_2},h_2}g_1g_2h_2^*
$$
as a consequence of the familiar ``commutation relation'' $h g = g (h^g)$ where, as above, we adopt the notation $h^g = ghg^{-1}$.
Physically, one can think of the elements of the group as creation operators and the elements of the dual as annihilation operators. 
The multiplication rule, then, expresses a natural commutation relation typically more exotic than the usual one for bosons ($aa^\dag=a^\dag a$) or fermions ($a a^\dag = -a^\dag a$).

\subsection{Quasi-triangularity and braiding}\label{sec:braiding}
 Quantum doubles of finite groups possess a further property of interest: they are \emph{quasi-triangular}
(or \emph{braided}). Specifically, there is an invertible element $R$ of $A \otimes A$ that imparts a ``near cocommutativity'' property on $\Delta$,
\begin{equation}
R\Delta(a)R^{-1}=T\Delta(a)\,,\label{eq:R-commutation}
\end{equation}
and satisfies the braiding relations on $\D(G)^{\otimes 3} = \D(G) \otimes \D(G) \otimes \D(G)$:
\begin{equation}
\label{eq:YB}
(\Delta\otimes \vec{1})(R)=R_{13}R_{23}\,,\qquad (\vec{1}\otimes\Delta)(R)=R_{13}R_{12}\,,
\end{equation}
where $R_{12}$ (and $R_{23}$, respectively) denotes the element $R$ acting on the first two (last two, respectively) tensor copies. (As above, $T: \alpha \otimes \beta \mapsto \beta \otimes \alpha$ is the exchange operator.)

The element $R$ is traditionally called an ``$R$-matrix,'' and immediately yields a solution to the Yang-Baxter equations: if we define $s_1=TR_{12}\otimes I$ (and likewise for $s_2$ and $s_3$), then it can also be shown from~\eqref{eq:YB} that
\begin{equation}\label{eq:double-braid}
s_1s_2s_1=s_2s_1s_2\,.
\end{equation}
In the algebra $\D(G)^{\otimes n}$, we generalize this notation by analogously defining $s_i$, for $1 \leq i < n$, as $I \otimes TR \otimes I$, the operator $TR$ acting on the $i$th and $i+1$st factors of $\D(G)$.

A generalization of the quantum double construction described above for finite groups can be applied to generic Hopf algebras (and their dual) to yield a quasi-triangular Hopf algebra with an explicit expression for the $R$ matrices. For the finite group double $\D(G)$, the $R$ matrix is given by
\[
R=\sum_g g\otimes g^* \in \D(G) \otimes \D(G)\,,
\]
where we embed $g\in G$ into $\D(G)$ as $\sum_h gh^*$ and $g^*$ as $eg^*$.

\paragraph{The braid group}
Recall that the braid group on $n$ strands, $B_n$, is an infinite, discrete group generated by the elements $\{ \sigma_1,  \ldots, \sigma_{n-1}\}$ with the following relations:
\begin{align}
\sigma_i\sigma_j &=\sigma_j\sigma_i &&\text{for  $|i-j|\geq 2$}\,, \nonumber \\
\sigma_i\sigma_{i+1}\sigma_i &=\sigma_{i+1}\sigma_i\sigma_{i+1} &&\text{for $i\in [1,n-1]$}\,. \label{eq:braid-YB}
\end{align}
With the natural topological realization of the braid group as equivalence classes of braids under ambient isotopy, the generator $\sigma_i$ can be naturally associated with the topological braid of Figure~\ref{fig:generator-topological} and the relations~\eqref{eq:braid-YB} correspond to the topological equivalence of Figure~\ref{fig:Yang-Baxter}.

\paragraph{Representations of the braid group via the quantum double} As the elements $s_i$ of~\eqref{eq:double-braid} above satisfy the relations~\eqref{eq:braid-YB}, it follows that if $\rho: \D(G) \rightarrow \GL(V)$ is a representation of $\D(G)$ we can obtain a representation $\tau = \rho^{\otimes n}$ of $B_n$ on $V^{\otimes n}$ by letting $\tau(\sigma_i)=\rho^{\otimes n}(s_i)$.

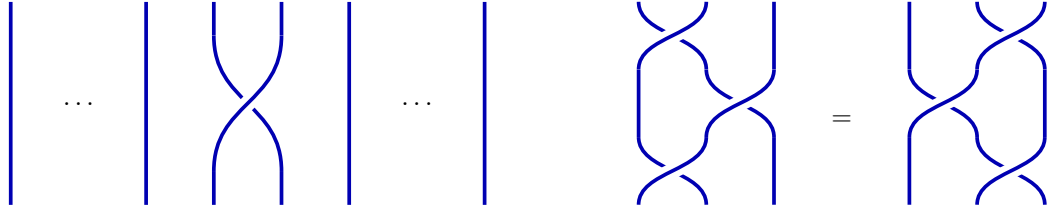
\begin{figure}
\begin{center}
\begin{subfigure}[b]{.45\textwidth}
\centering
\begin{tikzpicture}[style=stangle,scale=.6]

\draw[double] (-8,6) -- (-8,0);
\draw[double] (-6,3) node {$\ldots$};
\draw[double] (-4,6) -- (-4,0);
\draw[double] (-2,6) -- (-2,5);
\draw[double] (0,6) -- (0,5);
\draw[double] (-2,5) .. controls +(0,-2) and +(0,2) .. (0,1);
\draw[double] (0,5) .. controls +(0,-2) and +(0,2) .. (-2,1);
\draw[double] (-2,1) -- (-2,0);
\draw[double] (0,1) -- (0,0);
\draw[double] (2,6) -- (2,0);
\draw[double] (4,3) node {$\ldots$};
\draw[double] (6,6) -- (6,0);
\end{tikzpicture}
\smallskip
\caption{The topological realization of the element $\sigma_i$, involving strands $i$ and $i+1$.}
\label{fig:generator-topological}
\end{subfigure}
\quad
\begin{subfigure}[b]{.45\textwidth}
\centering
\begin{tikzpicture}[style=stangle,scale=.6]

\draw[double] (-2,4) .. controls (-2,3) and (0,3) .. (0,2);
\draw[double] (0,4) .. controls (0,3) and (-2,3) .. (-2,2);
\draw[double] (-2,2) -- (-2,0);

\draw[double] (0,2) .. controls (0,1) and (2,1) .. (2,0);
\draw[double] (2,2) .. controls (2,1) and (0,1) .. (0,0);
\draw[double] (2,4) -- (2,2);

\draw[double] (-2,0) .. controls (-2,-1) and (0,-1) .. (0,-2);
\draw[double] (0,0) .. controls (0,-1) and (-2,-1) .. (-2,-2);
\draw[double] (2,0) -- (2,-2);

\draw (4,0.5) node {$=$};


\draw[double] (8,4) .. controls (8,3) and (10,3) .. (10,2);
\draw[double] (10,4) .. controls (10,3) and (8,3) .. (8,2);
\draw[double] (6,4) -- (6,2);

\draw[double] (6,2) .. controls (6,1) and (8,1) .. (8,0);
\draw[double] (8,2) .. controls (8,1) and (6,1) .. (6,0);
\draw[double] (10,2) -- (10,0);

\draw[double] (8,0) .. controls (8,-1) and (10,-1) .. (10,-2);
\draw[double] (10,0) .. controls (10,-1) and (8,-1) .. (8,-2);
\draw[double] (6,0) -- (6,-2);
\end{tikzpicture}
\smallskip
\caption{The topological interpretation of the Yang-Baxter relations: $\sigma_i\sigma_{i+1}\sigma_i = \sigma_{i+1}\sigma_i\sigma_{i+1}$.}
\label{fig:Yang-Baxter}
\end{subfigure}
\caption{Topological realization of the braid group.}
\label{fig:braid-topological}
\end{center}
\end{figure}

\section{The Quantum Fourier transform over \texorpdfstring{$\D(G)$}{D(G)}}\label{Fourier}
\subsection{Background in representation theory}
We assume that the reader is familiar with the representation theory of finite groups and merely set down notation below, following~\citet{Serre1977linear}.

Let $G$ be a finite group. A representation $\rho$ of $G$ is a homomorphism $\rho: G \rightarrow \U(V)$, where $V$ is a finite-dimensional complex Hilbert space and $\U(V)$ denotes the unitary group on $V$. In this unitary case, it is convenient to work with a \emph{$G$-invariant} inner product, one for which $\langle \vec{u}, \vec{v}\rangle = \langle \rho(g) \vec{u}, \rho(g)\vec{v}\rangle$. By symmetrizing an arbitrary inner product $\langle \cdot, \cdot \rangle$ on $V$, one can define a new, $G$-invariant inner product by the rule
\[
\langle \vec{u}, \vec{v}\rangle' = \frac{1}{|G|} \sum_{g \in G} \langle \rho(g)\vec{u}, \rho(g)\vec{v}\rangle\,;
\]
thus we work under this assumption without loss of generality. 

To be concise, we will simply say that \emph{$V$ is a representation} of $G$ (omitting reference to the map $\rho$) and let $g\vec{v}$ denote $\rho(g) \vec{v}$, the linear action of an element $g \in G$ on a vector $\vec{v} \in V$. Two representations $V$ and $V'$ are isomorphic if there is a linear isomorphism $\phi: V \rightarrow V'$ for which $g\phi(\vec{v}) = \phi(g \vec{v})$ for all $g$ and $\vec{v}$. We say that a subspace $W$ of a representation $V$ is \emph{invariant} if $g\vec{w} \in W$ for all $g \in G$ and $\vec{w} \in W$. Such an invariant subspace is a representation itself under the restricted action. Of course, both $\{ 0\}$ and $V$ are invariant; if these are the only invariant subspaces, $V$ is \emph{irreducible}. When $V$ is not irreducible, there is a nontrivial invariant subspace $W \subset V$ and it is easy to check that $W^\perp$ is likewise invariant. This realizes $V$ as the direct sum of two, orthogonal invariant subspaces, each a representation of $G$. When a representation $\rho: G \rightarrow \U(V)$ can be decomposed in this way so that $V = W_1 \oplus W_2$, we write $\rho = \sigma_1 \oplus \sigma_2$, where $\sigma_i: G \rightarrow \U(W_i)$ is the restricted action on $W_i$. This process allows one to decompose any finite-dimensional representation into an orthogonal direct sum of irreducible representations. While this decomposition is not, in general, unique, the number of appearances of representations of each isomorphism class is determined uniquely.

For a finite group $G$, the vector space $\C[G]$ can be naturally given the structure of a representation by linearly extending the left multiplication map $g: h \mapsto gh$; this is a representation of degree $|G|$ called the \emph{regular} representation which we will alternately denote $\text{Reg}_G$. Likewise, the one-dimensional vector space $\C$ can be given the structure of a representation by linearly extending the rule $g: 1 \mapsto 1$; this is the \emph{trivial} representation. The regular representation $\C[G]$ has a remarkable, canonical decomposition into irreducible representations of $G$: each irreducible representation of $G$ appears with multiplicity equal to its dimension. We write
\begin{equation}
\label{eq:reg-decomp}
\C[G] = \bigoplus_{\rho} \rho^{\oplus d_\rho}\,,
\end{equation}
where this sum extended over all (isomorphism classes of) irreducible representations $\rho: G \rightarrow \U(V_\rho)$ and $d_\rho$ denotes $\dim V_\rho$. 

A (unitary) linear operator on $\C[G]$ that carries the basis $\{ g \mid g \in G\}$ to a basis consistent with the decomposition~\eqref{eq:reg-decomp} is called a \emph{Fourier transform} (over $G$). When convenient, we adopt the Dirac vector notation to avoid lengthy subscripts and write target basis for $\bigoplus_\rho \rho^{\oplus d_\rho}$ in the form $\ket{\rho, i, j}$, where $\rho$ is a representation of $G$, and $i, j \in \{1, \ldots, d_\rho\}$. Here we treat $j$ as indexing one of the $d_\rho$ copies of $\rho$ in $\C[G]$; thus, for each $\rho$ and $j$, the space spanned by $\{ \ket{\rho, i, j} \mid i \in \{1, \ldots, d_\rho\}\}$ is an irreducible subspace isomorphic to $\rho$.

Given a subgroup $H$ of $G$ and a representation $(\rho,V)$ of $H$, one can construct a representation of $G$ from $\rho$ by a process called \emph{induction} which plays an important role in the representation theory of $\D(G)$. In this subsection, we describe the construction in general. As a $G$ module, the induced representation can be neatly expressed as a tensor product (over the noncommutative ring $\C[H]$)  
\[
\C[G]\otimes_{\C[H]}V\,,
\]
where $\otimes_{\C[H]}$ indicates that this tensor product is $\C[H]$-linear, i.\,e., $gh\otimes \vec{v} = g\otimes \rho(h)\vec{v}$, for $h\in H$ and $\vec{v}\in V$. In order to give an explicit construction in terms of a basis, let us pick a transversal $T=\{t_1,t_2,\dots ,t_k\}$ for $H$ in $G$ (so that $k=|G|/|H|$) and fix a basis $\{\ket{v_1},\dots ,\ket{v_d}\}$ for the representation space $V$. The induced representation, which we denote $\uparrow_H^G\rho$ or simply $\uparrow^G\rho$ when $H$ is understood, acts on the vector space $\C[T] \otimes V$, where $\C[T] = \{ \sum_i a_i \ket{t_i} \mid a_i \in \C\}$ is the (Hilbert) space of formal $\C$-linear combinations of elements of $T$ with the inner product $\langle{t_i}\mid{t_j}\rangle = \delta_{ij}$. (The tensor product in this case is the conventional tensor product over $\C$.) In this basis, $(\uparrow_H^G \rho)(g)$ is given by linearly extending the rule
\begin{equation}
\label{eq:induced-action}
g: \ket{t} \otimes \ket{v} \mapsto \ket{t'} \otimes h\ket{v}
\end{equation}
where
$t' \in T$ and $h \in H$ are the unique elements for which $gt = t'h$.

There are two facts about induced representations which we use later. 
\begin{enumerate}
\item Induction is transitive: if $H\leq K\leq G$ then
  \begin{equation}
    [\uparrow_K^G (\uparrow_H^K \rho)] \cong \;\uparrow_H^G \rho\,.
  \end{equation}
\item The regular representation of $G$ is an induced representation from the trivial representation of the trivial subgroup $\{e\}$.
  \begin{equation}
    \text{Reg}_G=\uparrow_{\{e\}}^G \vec{1}\,.
  \end{equation}
\end{enumerate}
Combining the above two facts, we can write
\begin{equation}
  \text{Reg}_G= \uparrow_H^G (\uparrow_{\{e\}}^H \vec{1}) = \uparrow_H^G \text{Reg}_H\,.
\end{equation}

\subsection{\texorpdfstring{Irreducible representations and Fourier analysis over $\D(G)$}{Irreducible representations and Fourier analysis over D(G)}}

A \emph{representation} of a Hopf algebra $H$ is a representation of its algebra structure. Specifically, a representation $(\rho, V)$ of the Hopf algebra $H$ is a homomorphism
$$
\rho: H \rightarrow \End(V)\,.
$$
As above, the algebra $\D(G)$ can be given the structure of a representation using left-multiplication. The irreducible representations of $\D(G)$ have been worked out by \citet{DPR} in the physics literature and by \citet{Gould:qd-reps} in the mathematics literature. In the article~\cite{Gould:qd-reps}, Gould established that the algebra $\D(G)$ is semi-simple (when $G$ is a finite group), described the irreducible representations of $\D(G)$, and developed an associated character theory. In this section, we briefly describe his construction of the irreducible representations of $\D(G)$, indicate how the regular representation can be broken down into irreducible representations, and proceed to construct quantum circuits for the associated quantum Fourier transform.

To simplify the notation in what follows, we adopt Dirac notation for the algebras $\C[G]$ and $\C G$, treating them as spanned by the orthonormal bases $\{ \ket{g} \mid g \in G\}$ and $\{ \ket{h^*} \mid h \in G\}$, respectively. Writing $\ket{gh^\ast}$ as shorthand for the tensor product $\ket{g,h^\ast}$, the action of $\D(G)$ on the regular representation, spanned by the $\ket{gh^*}$, is
\begin{equation}
\label{eq:DG-action}
g^\prime (h^\prime)^\ast \ket{g,h^\ast} = \delta_{(h^\prime)^g,h}\ket{g^\prime g, h^\ast}\,.
\end{equation}
We see that $g' (h^\prime)^\ast$ takes $\ket{g,h^\ast}$ to zero unless $h$ and $h^\prime$ are conjugates (under the action of $g$); in the case where $g' (h^\prime)^\ast$ does not annihilate $\ket{g,h^\ast}$, it acts by the left $G$ action in the first coordinate. It follows immediately that each of the subspaces
\[
V_h=\text{span}(\{\ket{g,h^*},g\in G\})
\]
is invariant under the left action of $\D(G)$; this provides an initial orthogonal decomposition $\D(G) = \bigoplus_h V_h$ into invariant subspaces. Observe, also, that the action of $\sum_{h \in g} g h^\ast$ on $V_h$ is precisely left multiplication by $g$ on the first component; this gives $V_h$ the structure of a $G$-representation isomorphic to $\C[G]$. As the action of $\D(G)$ on $V_h$ is at least as rich as the action of $\C[G]$ on $V_h$, the irreducible subspaces of $V_h$ under the $\D(G)$ action are direct sums of irreducible subspaces under the $G$ action.\footnote{The map $g \mapsto \sum_h gh^*$ is a one-to-one algebra homomorphism from $\C[G]$ into $\D(G)$; thus any $\D(G)$ representation can be given the structure of a $G$ representation.} Remarkably, we will see that these representations arise directly as induced representations from the centralizer subgroup $Z(h) = \{ g \in G \mid gh = hg\}$.

Fixing an $h \in G$, let $T_h$ be a transversal for $Z(h)$ in $G$; then we may write each basis vector of $V_h$ uniquely:
\[
\ket{g,h^*}=\ket{tz , h^*}\,,
\]
where $t\in T_h$ and $z\in Z(h)$. In particular, we have the straightforward isometric isomorphism $\C[G] \equiv \C[T_h] \otimes \C[Z(h)]$ given by the map $\ket{g} \mapsto \ket{t} \otimes \ket{z}$. With this isomorphism, we shall work with the basis 
\begin{equation}
\label{eq:coset-basis}
\bigl\{ \ket{t, z, h^*} \bigr\}
\end{equation}
 (for $V_h$). Recall from the discussion of induction, above, that $\C[G] \cong\; \uparrow_{Z(h)}^G\text{Reg}_{Z(h)}$. Indeed, this $\ket{t, z, h^*}$ basis is precisely the basis described in the explicit formulation of induction given in~\eqref{eq:induced-action}.

It is easy to check that induction commutes with decomposition: when $H$ is a subgroup of $G$ and $\rho$ and $\tau$ are two $H$-representations we have a natural isomorphism between the two $G$-representations $\uparrow_H^G (\rho \oplus \tau)$ and $\uparrow_H^G \rho \;\oplus \uparrow_H^G \tau$. According to~\eqref{eq:reg-decomp}, the group algebra $\C[Z(h)]$ decomposes into the orthogonal sum $\bigoplus_{\rho} \rho^{\oplus d_\rho}$, in which each irreducible representation $\rho$ of $Z(h)$ appears $d_\rho = \dim \rho$ times. We conclude that $\C[G] = \bigoplus_\rho \uparrow_{Z_h}^G \rho$, this sum extended over all irreducible representations of $Z(h)$. We remark that when $\rho$ is an irreducible representation of $Z(h)$, the representation $\uparrow_{Z(h)}^G \rho$ is not, in general, irreducible as a $G$ representation. We shall see, however, that the natural $\D(G)$ action on $\uparrow_{Z(h)}^G \rho$ is irreducible; in fact, all irreducible $\D(G)$ representations arise in this way. Carrying out the Fourier transform on the component corresponding to $Z(h)$, we finally arrive at the basis
\(\ket{t,\rho,i,j,h^\ast}\); for convenience, we reorder the components and work with the basis
\begin{equation}
\bigl\{\ket{h^\ast,\rho,j,t,i}\bigr\}\,.
\label{eq:DG-F-basis}
\end{equation}
For a given $\rho$, $h$, and $j$, the vectors $\ket{h^\ast,\rho,j,t,i}$ span an induced representation of $\rho$ to $G$; we denote this space as $V_{h,\rho,j}$. The action of $G$ on these vectors is that of $\uparrow_{Z(h)}^G \rho$; in particular, $G$ leaves this space invariant. Though the space is not irreducible, in general, under the $G$ action, it is under the richer action afforded by $\D(G)$. In light of~\eqref{eq:DG-action}, the action of any $(h^\prime)^\ast$ preserves this space; specifically,
\begin{itemize}
\item
if $h^\prime$ is not a conjugate of $h$, it annihilates $V_h$ and hence $V_{h, \rho, j}$;
\item if $h'$ is conjugate to $h$, the action of $h^*$ on the space $V_{h}$ is the projection onto the span of
\[
\{ \ket{g, h^*} \mid h' = h^g\}\,.
\]
Note, however, that $h^{zg} = h^g$ for any $z \in Z(h)$ and hence $\{ g \mid h' = g^h\}$ is a left coset of $Z(h)$. Evidently, there is an element $t' \in T_h$ for which $h^* \ket{t, z, h^*} = \delta_{t, t'} \ket{t, z, h^*}$ (in the basis of~\eqref{eq:coset-basis} above) and we conclude that on the space $V_{h, \rho, j}$ the action of $h^*$ projects onto the span of the vectors
\[
\{ \ket{h, \rho, j, t', i} \mid i \in \{ 1, \ldots, d_\rho\}\}\,.
\]
We let $V_{h, \rho, j}^{t'}$ denote this subspace.
\end{itemize}
Note that $V_{h, \rho, j} = \bigoplus_{t \in T_h} V_{h, \rho, j}^t$, an orthogonal direct sum. To see that this space is irreducible under the $\D(G)$ action, consider a $\D(G)$-invariant subspace $W \subset V_{h, \rho, j}$. As described above, the $\D(G)$ action contains the orthogonal projection operators onto each $V_{h, \rho, j}^t$; it follows that 
\[
W = \bigoplus_{t \in T_h} W_t\quad \text{where each $W_t = W \cap V_{h, \rho, j}^t$.} 
\]
Now we turn our attention to the $G$ action. For concreteness, we assume that $1 \in T_h$ is the representative for the coset $Z(h)$; then note that the subspace $V_{h, \rho, j}^1$ is invariant and, moreover, irreducible under the $Z(h)$ action as it is isomorphic to the irreducible $Z(h)$-representation $\rho$. It follows that $W_1$ is either $0$ or $V_{h, \rho, j}^1$. As $V_{h, \rho, j}$ has the structure of an induced representation (it is isomorphic to $\Ind_{Z(h)}^G \rho$), the $G$ action is transitive on the spaces $V_{h, \rho, j}^t$: in particular, for any $t_1, t_2 \in T_h$, there is an element $g \in G$ so that $g(T_{h, \rho, j}^{t_1}) \subset T_{h, \rho, j}^{t_2}$ and, hence, $g(W_{t_1}) \subset g(W_{t_2})$. We conclude that $\dim W_{t_1} = \dim W_{t_2}$ for each $t_1, t_2$ and, finally, that $W = 0$ or $V_{h, \rho, j}$, as desired.

This yields an orthogonal decomposition
\[
\D(G) = \bigoplus_{h \in G} \bigoplus_{\rho \in \widehat{Z(h)}} \bigoplus_{j \in \{ 1, \ldots, d_\rho\}} V_{h, \rho, j}
\]
into irreducible subspaces. Here the notation $\widehat{Z(h)}$ denotes the collection of irreducible representations of $Z(h)$ (upto isomorphism). As $\D(G)$ is semisimple~\cite{Gould:qd-reps}, all irreducible representations appear in this decomposition. This decomposition will be sufficient to describe the quantum Fourier transform over $\D(G)$.

 To complete the representation-theoretic picture, however, we discuss a few more details. It is not difficult to show that the $\D(G)$ representations $V_{h, \rho, j}$ and $V_{h', \rho', j'}$ are isomorphic whenever the pair $(h, \rho)$ and $(h',\rho')$ are conjugate in the sense that there is a group element $g$ so that $h' = h^g$ and the $Z(h')$ representation $(\rho')^g: z \mapsto \rho'(z^g)$ is isomorphic to $\rho$.
 
In order to describe an explicit action of $\D(G)$ on its constituent irreducible representations, note that for a transversal we can pick elements $k_g$ such that $(k_g)^{-1} g k_g =h$, where $g$ is a conjugate of $h$. The the action of $\D(G)$ can then be written as
\[
gh^\ast \ket{k_{g^\prime}, v} = \delta_{h,g^\prime}\ket{k_{gg^\prime g^{-1}}, \rho(k_{gg^\prime g^{-1}}^{-1}g^\prime k_{g^\prime})v} \,.
\]
For notational simplicity, we denote the transversal element $k_g$ by $g$ in the first register and write the above equation as
\begin{equation}\label{R:InsideIrrep}
g h^\ast \ket{g^\prime, v} = \delta_{h,g^\prime}\ket{g g^\prime g^{-1}, \rho(k_{g g^\prime g^{-1}}^{-1}g^\prime k_{g^\prime})v}\,.
\end{equation}

Finally, a remark about conjugate (contragredient) representations over $\D(G)$. If $\rho: \D(G) \rightarrow \GL(V)$ is a representation of $\D(G)$, its conjugate representation $\rho^*: \D(G) \rightarrow \GL(V^*)$ is defined by the rule
$\rho^*(\alpha)f: v \mapsto f(\rho(S\alpha) v)$. It follows that the conjugate of the irreducible representation $\Ind_{Z(h)} \rho$ is the representation $\Ind_{Z(h^{-1})} \rho^*$. Self-dual representations will play a special rule in our applications to knot theory; see Section~\ref{sec:plat}.

\subsection{The quantum Fourier transform over \texorpdfstring{$\D(G)$}{D(G)}}
The quantum Fourier transform (QFT) over $\D(G)$ is a unitary matrix which transforms the $\{\ket{g,h^\ast}\}$ basis to the $\{ \ket{h^\ast,\rho,j,t,i} \}$ basis (using the notation defined in the previous section). The discussion above yields an explicit decomposition of $\D(G)$ into irreducible representations assuming such a decomposition for each group algebra $\C[Z(h)]$. It follows that one can efficiently carry out the quantum Fourier transform over $\D(G)$ if
\begin{itemize}
\item given $h$, one can efficiently carry out the QFT over $Z(h)$; and
\item given $h$ and $g$, one can efficiently express $g = t z$, where $z \in Z(h)$ and $t$ is an element of a fixed transversal $T_h$ of $Z(h)$. (Here ``fixed'' means that $T_h$ may depend on $h$, but not on $g$.)
\end{itemize}
We remark that the first condition is stronger than merely insisting that there be an efficient quantum Fourier transform over the group $Z(h)$: the circuit must be efficiently computable from $h$ in time $\polylog(|G|)$. Note that the algorithm promised in the second condition implicitly determines a transversal $T_h$ for each $Z(h)$. In the following section, we show that these conditions are satisfied for $S_n$, the symmetric group.

With such algorithms, the QFT itself is straightforward:
\begin{enumerate}
\item Convert the basis $\ket{g,h^\ast}$ to $\ket{t,z,h^\ast}$, where $t\in T_h$ and $z\in Z(h)$.
\item Rewriting it as $\ket{h^\ast,t,z}$, apply the QFT over $Z(h)$ on the third register conditioned on $h^\ast$ in the first register. We obtain the basis $\ket{h^\ast,t,\rho,i,j}$, where $\rho\in\widehat{Z(h)}$. Reordering again yields the basis of~\eqref{eq:DG-F-basis} above.
\end{enumerate}
This gives us the QFT over $\D(G)$.

\revise{
Clearly, efficient QFT over each $Z(h)$ is sufficient to construct the QFT over $\D(G)$. Now, we show that this condition is also necessary. More precisely, we show that if one can perform a QFT over $\D(G)$, then one can use it to implement a QFT over any $Z(h)$. Suppose we have a state $\sum_z f(z)\ket{z}$, where $z\in Z(h)$, for some $h\in G$. We would like to obtain the state $\sum_{\rho ,i,j} \hat{f}(\rho,i,j) \ket{\rho,i,j}$, where $\hat{f}$ is the Fourier transform of $f$. First, we embed this state into $\D(G)$ as $\sum_z f(z)\ket{h^\ast,t,z}$, for some $t\in T_h$. Now we can perform a QFT over $\D(G)$ and obtain $\sum_{\rho ,i,j,t,h} \hat{g}(\rho,i,j,t,h)\ket{h^\ast,\rho,t,i,j}$, where $g(h^\prime,t^\prime,z)=f(z) \delta_{t,t^\prime} \delta_{h,h^\prime}$. Since $h$ and $t$ are unaffected, we have that $\hat{g}(\rho,i,j,t,h)=\hat{f}(\rho,i,j) \delta_{t,t^\prime} \delta_{h,h^\prime}$. Now throwing away the registers containing $t$ and $h^\ast$, we obtain the QFT over $Z(h)$. This shows that performing a QFT over each $Z(h)$ is necessary and sufficient to perform a QFT over $\D(G)$. Finally, note that having an efficient QFT over each centralizer $Z(h)$ implies that one has an efficient QFT over $G$ since $Z(e)=G$, where $e$ is the identity element of $G$.
}

\subsubsection{An efficient QFT over \texorpdfstring{$\D(S_n)$}{the quantum double of the symmetric group}}

In light of the discussion above, we show that one can uniformly compute the quantum Fourier transform over centralizers in $S_n$ and, additionally, uniformly factor along a canonical transversal of each centralizer. It follows that we can compute the QFT over $\D(S_n)$.

First, let us determine the centralizer of an element $\sigma$ of $S_n$ consisting of $c_1$ one cycles, $c_2$ two cycles, etc. For an element $\tau \in Z(\sigma)$, we have $\tau \sigma = \sigma \tau$ and hence $\tau \sigma \tau^{-1} = \sigma$: $\tau$ fixes $\sigma$ under conjugation. Recall that the conjugation action by $\tau$ is given by the permutation action of $\tau$ on the labels of the cycle representation of $\sigma$; thus, for each $k$, the $k$-cycles of $\sigma$ are fixed under the action $(a_1\,\ldots\,a_k) \mapsto (\tau(a_1)\,\ldots\,\tau(a_k))$. At a coarser level, the permutation $\tau$ fixes, as a set, the elements of 
\[
W_k = \{ t \in \{1, \ldots, n\} \mid \text{$t$ lies in a $k$-cycle of $\sigma$}\}\,,
\] 
and it follows that $Z(\sigma)$ is the direct product $Z(\sigma_1) \times \cdots \times Z(\sigma_n)$, where $\sigma_k$ is the product of the $c_k$ cycles of length $k$ appearing in $\sigma$ and $Z(\sigma_k)$ is the centralizer of $\sigma_k$ (in the subgroup of permutations on $W_k$). Observe that when $(\tau(a_1) \,\ldots\, \tau(a_t)) = (a_1\, \ldots\, a_t)$, the element $\tau$ (restricted to $\{ a_1, \ldots, a_t\}$) lies in the cyclic subgroup generated by $(a_1, \ldots, a_t)$. It follows that on $W_k$ the element $\tau$ can be written as the product of two permutations, one which cyclicly permutes each $k$-cycle, and one which permutes the $k$-cycles. This subgroup of permutations on $W_k$ is isomorphic to the wreath product $S_{c_k} \wr \Z_k = S_{c_k} \ltimes \Z_k^{c_k}$ and we conclude that
\[
Z(\sigma) \cong \bigoplus_{k=1}^n S_{c_k} \ltimes \Z_k^{c_k}\,.
\]

\paragraph{Coset factorization}
Let $H$ be a subgroup of $G$ and $T$ a transversal of $H$ in $G$ (that is, a set containing one representative from each left coset of $H$ in $G$). \emph{$(H,G)$-coset factorization} is the problem of expressing an arbitrary element $g$ of $G$ as a product $th$, where $h \in H$ and $t \in T$. Observe that if $K < H < G$ is a tower of subgroups and $T_K$ and $T_H$ are transversals of $K$ in $H$ and $H$ in $G$, respectively, then $T_K T_H$ is a transversal of $K$ in $G$ and $(K,G)$-coset factorization along this transversal reduces to $(K,H)$-coset factorization and $(H, G)$-coset factorization.

The induced bases we use to describe the Fourier transform above rely on coset factorization of an element along arbitrary, but fixed transversals $T_\sigma$ of the various centralizers $Z(\sigma)$. To see that these can be computed effectively in the symmetric groups, let $\sigma \in S_n$ and 
\[
Z(\sigma) \cong \bigoplus_{k=1}^n S_{c_k} \ltimes \Z_k^{c_k}\,,
\]
its centralizer, as discussed above. Our goal is to effectively decompose an arbitrary element $\pi \in S_n$ as a product $\pi = t z$, where $z \in Z(\sigma)$ and $t$ lies in a transversal $T_\sigma$ (which may be determined implicitly by the algorithm). We consider the subgroup chain
\newcommand{\subsetnumbered}[1]{\subset_{(#1)}}
\[
\bigoplus_k S_{c_k} \ltimes \Z_k^{c_k} \quad \subsetnumbered{1} \quad \bigoplus_k S_{c_k} \ltimes S_k^{c_k} \quad \subsetnumbered{2} \quad \bigoplus_k S_{c_k \times k} \quad \subsetnumbered{3} \quad S_n\,.
\]
It is an easy exercise to identify natural transversals for the inclusions $\subsetnumbered{1}$ and $\subsetnumbered{3}$ that permit efficient factorization. As for the factorization corresponding to $\subsetnumbered{2}$, it suffices to handle individual terms in the direct sum of the form
\[
S_k \ltimes S_\ell^k \quad \subset \quad S_{k\ell}\,.
\]
For this purpose, we consider an associated action of $S_{k\ell}$. We say that a family of sets $(A_1, \ldots, A_\ell)$ is a $(k, \ell)$-partition if their disjoint union is the set $\{1, \ldots, k\ell\}$ and each $A_i$ has size $k$. Then consider the action of $S_{k\ell}$ on the set
\[
X_{k,\ell} = \bigl\{ \{ A_1, \ldots, A_k\} \mid \text{$(A_1, \ldots, A_k)$ is a $(k,\ell)$-partition of $\{1, \ldots, k\ell\}$} \bigr\}\,.
\]
The stabilizer of the element
\[
x_0 = \{ \{1, \ldots, k\}, \{ k+1, \ldots, 2k\}, \ldots\}
\]
is precisely a subgroup of the form $S_\ell \ltimes S_k^\ell$ so we identify a left transversal of $S_\ell \ltimes S_k^\ell$ in $S_{k\ell}$ by identifying, for each element $x$ of $X_{k,\ell}$, a permutation that carries $x_0$ to $x$. This can be obtained by ``sorting'' the sets of $x$ according to their smallest element, ``sorting'' the elements of each individual set, and selecting the permutation that carries $x_0$, expressed in the order above, to $x$, in this sorted order. Factorization along this transversal is straightforward by identifying the result of a permutation $\pi$ on $x_0$.

\paragraph{QFTs over the centralizers}
It is enough to describe how to perform a QFT over wreath products of the form $W=\Z_k \wreath S_\ell = (\Z_k)^{\ell} \rtimes S_\ell$. 

The irreps of such groups are easily obtained using Clifford theory (see~\cite{CurtisAndReiner}). Pick an irrep of $(\Z_k)^\ell$, say $\omega=(\omega_1,\dots,\omega_\ell)$ where each $\omega_i$ is an irrep of $\Z_k$, and consider the action of the symmetric group $S_\ell$ on the components $\omega_i$. Define $S_\omega$ to be the subgroup of $S_\ell$ that fixes this irrep, i.\,e., the subgroup whose elements permute components which are the equal. Thus $S_\ell$ is a Young subgroup, equal to a product of symmetric groups (each of which acts on a collection of equal indices). If $\lambda$ is an irrep of $S_\omega$, then $\omega\otimes\lambda$ is an irrep of $W_\omega=(\Z_k)^\ell\rtimes S_\omega$. Clifford theory asserts that the irreps of the wreath product are of the form
\[
\uparrow_{W_\omega}^W(\omega\otimes\lambda)\,.
\]
$W_\omega$ is traditionally called the \emph{inertia group} of the irrep $\omega$. 

The QFT over this group can now be constructed using this structure of the irreps. We need to construct a transformation which takes us from the basis $\ket{(z_1,\dots,z_\ell),\pi}\in W$ to the basis $\ket{t,\omega,\lambda}$, where $t$ is an element of a fixed transversal for $W_\omega$ in $W$. In order to do this, we first perform a QFT over $(\Z_k)^\ell$ to obtain the basis $\ket{\omega,\pi}$. Now conditioned on $\ket{\omega}$, we re-write this basis as $\ket{\omega, t,\pi^\prime}$, where $t$ is an element of a transversal for this Young subgroup and $\pi^\prime$ is an element of $W_\omega$. Again, conditioned on $\omega$, we perform a QFT over $W_\omega$ on the third register to obtain the basis $\ket{\omega,t,\lambda,i,j}$, where $\lambda$ is an irrep of $S_\omega$ and $i$ and $j$ label its rows and columns. This step may be carried out by Beals's algorithm~\cite{Beals:1997:QCF:258533.258548} for the QFT over the symmetric group. This gives the required basis.

\section{Algorithms for approximating link invariants}\label{algorithms}
\subsection{Link invariants}
A \emph{knot} is a closed (smooth) curve in $\R^3$ with no self intersections; a \emph{link} is a finite collection of non-intersecting knots. We study knots and links upto deformation; in several cases, we work with oriented variants. In particular, recall that a continuous map $D: \R^3 \times [0,1] \rightarrow \R^3$ is an \emph{isotopy} of $\R^3$ if each $D_t(x) = D(x,t)$ is a homeomorphism and $D_0$ is the identity map. We then identify two knots (or links) if there is an isotopy that carries the first onto the second; in this case we say that they are \emph{ambient isotopic}. 

A \emph{link invariant} is association of links to algebraic objects (numbers, groups, modules, etc.) that respects the equivalence relation of ambient isotopy: equivalent links must be associated with the same object. One well-studied framework for studying links is to represent them as certain canonical ``closures'' of braids and explore the algebraic properties of braids that preserve link equivalence in this representation. We consider two such closures here: trace closure and plat closure.

\subsubsection{Trace closure}
Trace closure of a braid is obtained by joining the ends (top to bottom) as shown in Figure~\ref{fig:trace-closure}; Alexander's theorem~\cite{Alexander} asserts that any link can be obtained in this fashion. Of course, a given link can be represented as the closure of many different braids. 
\begin{figure}
\centering
\begin{subfigure}[b]{0.3\textwidth}
\centering
\begin{tikzpicture}[style=stangle,scale=.2]
\draw[draw=red,semithick,scale=0.75] (-10,-5) rectangle (10,5);

\draw[double] (-7,4) -- (-7,6);
\draw[double] (-7,-4) -- (-7,-6);
\draw[double] (-9,6) -- (-9,-6);
\draw[double] (-7,-6) .. controls (-7,-8) and (-9,-8) .. (-9,-6);
\draw[double] (-7,6) .. controls (-7,8) and (-9,8) .. (-9,6);

\draw[double] (-5,4) -- (-5,6);
\draw[double] (-5,-4) -- (-5,-6);
\draw[double] (-11,6) -- (-11,-6);
\draw[double] (-5,-6) .. controls (-5,-10) and (-11,-10) .. (-11,-6);
\draw[double] (-5,6) .. controls (-5,10) and (-11,10) .. (-11,6);

\draw[double] (7,4) -- (7,6);
\draw[double] (7,-4) -- (7,-6);
\draw[double] (-21,6) -- (-21,-6);
\draw[double] (7,-6) .. controls (7,-14) and (-21,-14) .. (-21,-6);
\draw[double] (7,6) .. controls (7,14) and (-21,14) .. (-21,6);

\draw[basic] (-3,6) circle (4pt);
\draw[basic] (0,6) circle (4pt);
\draw[basic] (3,6) circle (4pt);
\draw[basic] (6,6) circle (4pt);

\draw[basic] (-3,-6) circle (4pt);
\draw[basic] (0,-6) circle (4pt);
\draw[basic] (3,-6) circle (4pt);
\draw[basic] (6,-6) circle (4pt);

\draw[basic] (-17,0) circle (4pt);
\draw[basic] (-20,0) circle (4pt);
\draw[basic] (-23,0) circle (4pt);
\draw[basic] (-26,0) circle (4pt);

\draw (0,0) node {\textsl{braid}};
\end{tikzpicture}
\caption{The trace closure.}\label{fig:trace-closure}
\end{subfigure}\quad
\begin{subfigure}[b]{0.3\textwidth}
\centering
\begin{tikzpicture}[style=stangle,scale=.2]

\draw[draw=red,semithick,scale=0.75] (-10,-5) rectangle (10,5);
\draw[double] (-7,4) -- (-7,6);
\draw[double] (-5,4) -- (-5,6);
\draw[double] (-7,6) .. controls (-7,8) and (-5, 8) .. (-5,6);

\draw[double] (-7,-4) -- (-7,-6);
\draw[double] (-5,-4) -- (-5,-6);
\draw[double] (-7,-6) .. controls (-7,-8) and (-5, -8) .. (-5,-6);

\draw[double] (7,4) -- (7,6);
\draw[double] (5,4) -- (5,6);
\draw[double] (7,6) .. controls (7,8) and (5, 8) .. (5,6);

\draw[double] (7,-4) -- (7,-6);
\draw[double] (5,-4) -- (5,-6);
\draw[double] (7,-6) .. controls (7,-8) and (5, -8) .. (5,-6);

\draw[basic] (-3,6) circle (4pt);
\draw[basic] (0,6) circle (4pt);
\draw[basic] (3,6) circle (4pt);

\draw[basic] (-3,-6) circle (4pt);
\draw[basic] (0,-6) circle (4pt);
\draw[basic] (3,-6) circle (4pt);

\draw (0,0) node {\textsl{braid}};
\end{tikzpicture}
\vspace{9mm}
\caption{The plat closure.}
\label{fig:plat-closure}
\end{subfigure}
\caption{Braid closures.}
\end{figure}
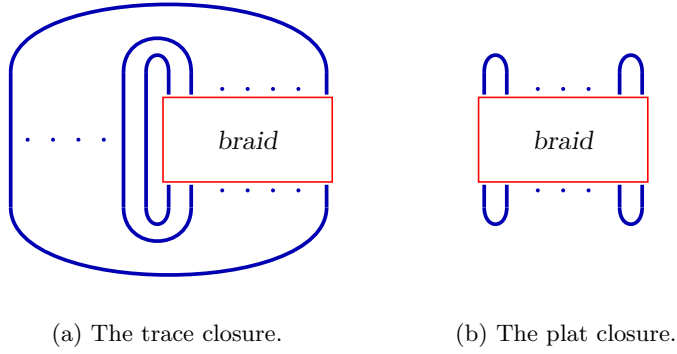

It has been shown by Markov~\cite{Markov} that braids that yield equivalent knots under the trace closure are related by the \emph{Markov moves}. These moves are
\begin{align*}
\theta\eta&\longleftrightarrow \eta\theta, &&\theta,\eta\in B_n\,, \text{and} \\
\theta&\longleftrightarrow \theta \sigma_{n-1}^{\pm 1}, &&\theta\in B_{n-1}\subset B_n\,.
\end{align*}
A \emph{Markov trace} is a map $\phi:B_n\rightarrow \C$ that is well-behaved with respect to these moves:
\[
\phi(\theta\eta)=\phi(\eta\theta)\,, \quad
\phi(\theta \sigma_{n-1})=z\phi(\theta)\,,  \quad
\phi(\theta \sigma_{n-1}^{-1})=\bar{z}\phi(\theta)\,,
\]
where $z=\phi(\sigma_{n-1})$ and $\bar{z}=\phi(\sigma_{n-1}^{-1})$. With such a trace map, one can construct a link invariant by defining, for any braid $\theta$ that realizes the link, the quantity
\begin{equation}L(\theta)=(z\bar{z})^{-(n-1)/2} \left(\frac{\bar{z}}{z}\right)^{e(\theta)/2}\phi(\theta) \,,
\end{equation}
where the linking number, $e(\theta)$, is the sum of the exponents of the generators $\sigma_i$ appearing in $\theta$.

It was shown by~\citet{TsohantjisG:qd-links} that the quantum doubles of finite groups yield link invariants by this approach. Specifically, let $(\rho, h)$ be an irrep of $\D(G)$ acting on the Hilbert space $V$ and let $\tau$ denote the natural representation of the braid group $B_n$, described in Section~\ref{sec:braiding} above, on the space $V^{\otimes n}$. Then the quantity
\begin{equation}
L_{\rho,h}(\theta)=(d_{\rho,h})^{-1} \langle h\rangle_\rho^{-e(\theta)}\tr(\tau(\theta))
\end{equation}
is a link invariant, where $\langle h\rangle_\rho=\chi_{\rho}(h)/d_\rho$ and $\chi_\rho, d_\rho$ are the character and dimension of the irrep $\rho$ of $Z(h)$. 

\subsubsection{Plat closure}\label{sec:plat}
A related method to represent links by braids is the \emph{plat closure}. This is defined on braids with an even number of strands where one takes pairs of adjacent strands on the top (and bottom) and joins them as in Figure~\ref{fig:plat-closure}. Similar to trace closure, it can be shown that any link can be represented as the plat closure of some braid.

\citet{Birman:Plat} proved an analogue of Markov's theorem for the plat closure.
\begin{theorem}[{\citet{Birman:Plat}}]
Given two braids in $B_{2n}$ with the same plat closure, there exists a sequence of moves from one to the other such that each move is one of the following:
\begin{itemize}
\item(Double coset move) $\theta\longleftrightarrow h_1\theta h_2$, where $h_1,h_2 \in H_{2n}$ (defined below) and,
\item(Stabilization move) $\theta\longleftrightarrow \sigma_{2n}\theta$.
\end{itemize}
Here, $H_{2n}$ denotes the \emph{Hilden group}, the subgroup of the braid group $B_{2n}$ generated by
\[
\sigma_1\,,\quad \sigma_2\sigma_1^2\sigma_2, \quad\text{and} \quad \sigma_{2i}\sigma_{2i-1}\sigma_{2i+1}\sigma_{2i}\,, \quad \forall i\in\{1,\dots ,n-1\}\,.
\]
\end{theorem}

The quantum double of a finite group $G$ likewise yields link invariants by this approach. Consider the $2n$-fold tensor power of an irrep $\Lambda$ for which $\bar{\Lambda}=\Lambda$, i.\,e., $\Lambda$ is its own conjugate. For adjacent pairs $\Lambda^{\otimes 2}$, define the state
\begin{equation}
\label{phi}
\ket{\Phi} = \frac{1}{\sqrt{d_{\Lambda}}}\sum_{v}\ket{v,v}\,,
\end{equation}
where the sum runs over all the elements of an orthonormal basis of $V_{\Lambda}$ and $d_\Lambda$ is its dimension. This maximally entangled state is the copy of the trivial representation in the tensor product. Now, define the state
\begin{equation}\label{alpha}
\ket{\alpha} = \ket{\Phi}^{\otimes n}
\end{equation}
and, given a braid $\theta \in B_{2n}$ that realizes the desired link under the Plat closure, let the representation of $\theta$ in the above tensor power be $\tau_{\Lambda}(\theta)$. The plat closure yields a link invariant:
\begin{equation}\label{Plat}
\Pl_{\Lambda}(\theta) = \bra{\alpha}\tau_{\Lambda}(\theta)\ket{\alpha}\,.
\end{equation}
Physically, one can view the plat closure as follows: particle-antiparticle pairs are created, their world lines are braided and then they are annihilated. The plat closure is the amplitude of obtaining the vacuum at the end of this process. One also requires that the particle be its own antiparticle since we need the representation to be its own conjugate.

\subsection{Quantum algorithms for link invariants}\label{QuantumAlg}
In this section, we present quantum algorithms to additively approximate the trace and plat closures of braids in some irrep $\Lambda$ of $\D(G)$. We use the QFT over $\D(G)$ to accomplish this: critically, the $R$ matrices can be implemented efficiently in the regular, combinatorial representation so we may use the QFT over $\D(G)$ to implement it inside the irreps. Then we take the trace over the irrep of choice. We describe this in more detail below. First, we need the definition of additive approximation (see~\citet{ApproximateCounting}).
\begin{definition}
Given any function $f:D\rightarrow \C$ and a normalization $u:\Z^{+}\rightarrow\R^{+}$, an \textbf{additive approximation} for the pair $(f,u)$ is a probabilistic algorithm which given any $x\in D$ and $\epsilon>0$ produces an output $g(x)$, such that
\[
\Pr[|f(x)-g(x)|>\epsilon u(|x|)]<1/4\,,
\]
in time polynomial in $|x|$ and $\epsilon^{-1}$.
\end{definition}

In order to approximate the trace closure of some braid $\theta$ written as a word in the generators $\sigma_i$, we first construct the representation of $\theta$ in $\D(G)^{\otimes n}$ by implementing the $R$ matrices in $\D(G)$. In the case of finite group doubles, the $R$ matrix is a unitary operator and, in fact, a permutation operator in the standard basis. Recall that it is defined as
\[
R=\sum_g g\otimes g^\ast \,;
\] 
thus its action on a basis element $\ket{g_1h_1^\ast}\otimes \ket{g_2h_2^\ast}$ is given by
\[
R\ket{g_1h_1^\ast,g_2h_2^\ast} = \sum_{g} \ket{gg_1h_1^\ast, g_2(g^{g_2}h_2)^\ast} = \ket{h_2^{g_2^{-1}}g_1h_1^\ast, g_2h_2^\ast} .
\]
Therefore, the action of $\sigma_1=TR$ is
\[
\sigma_1\ket{g_1h_1^\ast,g_2h_2^\ast} = \ket{g_2h_2^\ast,h_2^{g_2^{-1}}g_1h_1^\ast} .
\]
This operator is a left multiplication by a group element followed by a swap which can be efficiently implemented. 


In order to take the trace, we must be able to produce random basis vectors in the irrep $\Lambda$ embedded in $V_{\D(G)}$. The basis vectors inside the irrep $(h,\rho)$ can be constructed using the fact that it is an induced representation. If $\{w_1,\dots ,w_d\}$ are a basis for $V_\rho$, then $\{t_1\otimes w_1,\dots t_k\otimes w_d\}$ is a basis for $V_{\rho,h}$, where $t_i\in T_h$ as before. Assuming that we can embed the vectors $\ket{w_i}$ in the vector space $V_{Z(h)}$, then the vectors $\ket{h^\ast,t_j,w_i}$ are basis vectors of the irrep $(h,\rho)$ embedded in $V_{\D(G)}$. Now, using the Hadamard test, we can find the inner product
\[
\bra{v_i}F\tau_{\Lambda}(\theta) F^{-1}\ket{v_i}\,,
\]
for any given $v_i$, where $F$ is the QFT over $\D(G)$. The Hadamard test is a standard tool in quantum computing used to estimate the matrix entries of a unitary operator. The description below follows \cite{WY} (see also \cite{AJL}).
\begin{description}
\item[Hadamard test] Suppose that $U$ is a unitary operator that can be implemented efficiently, i.\,e., in $O(\text{poly}(n))$ time and $\ket{\psi}$ is a pure state on $n$ qubits which can be prepared efficiently, then one can efficiently sample from random variables $X,Y\in \{-1,1\}$, where
\[
\mathbb{E}[X+iY] = \bra{\psi}U\ket{\psi}\,.
\]
\end{description}
Finally, the Chernoff bound shows that choosing $v_i$ uniformly at random from the orthonormal basis and computing the inner product $\bra{v_i}F\tau_{\Lambda}(\theta) F^{-1}\ket{v_i}$ is enough to give an additive approximation.
\begin{description}
\item[Chernoff bound] If $\{X_1,\dots ,X_k\}$ are real-valued random variables such that $|X_i|\leq 1$ and $\mathbb{E}[X_i]=\mu$ for all $i$, then
\[
\Pr\left[\left|\frac{1}{k}\sum_{i=1}^k X_i - \mu \right|>\epsilon\right]\leq 2\exp\left(-\frac{k\epsilon^2}{4}\right)\,.
\]
\end{description}
In order to apply this to complex numbers $Z_j=X_j+iY_j$, we need the union bound along with the Chernoff bound.
\[
\Pr\left[\left|\frac{1}{k}\sum_{i=1}^k Z_i - \mathbb{E}Z_i \right|>\epsilon\right] \leq \Pr\left[\left|\frac{1}{k}\sum_{i=1}^k X_i - \mathbb{E}X_i \right|>\frac{\epsilon}{\sqrt{2}}\right] +\Pr\left[\left|\frac{1}{k}\sum_{i=1}^k Y_i - \mathbb{E}Y_i \right|>\frac{\epsilon}{\sqrt{2}}\right] \leq 4\exp \left(-\frac{k\epsilon^2}{8}\right)\,.
\]
To determine $k$ (the number of times we need to run the algorithm), let $Z_i=X_i+i Y_i$ where $X_i$ and $Y_i$ are obtained from the Hadamard test. Let $\hat{Z}_i=(d_{\rho,h})^{-1} \langle h\rangle_\rho^{-e(\theta)}Z_i$. Then for any $\epsilon$, the probability that the average of $\hat{Z}_i$ over $k$ trials differs from the trace closure by more than $\epsilon$ is less than $4\exp(-k\epsilon^2/8)$. In order to make this quantity less than $1/4$, we need $k> 32\log 2/\epsilon^2$.

The algorithm to additively approximate plat closure is similar. We must prepare the state $\ket{\alpha}$ from \eqref{alpha} efficiently. In order to do this, it is enough to prepare the state $\ket{\Phi}$ efficiently. Once we do this, we can perform the Hadamard test to estimate the inner product $\bra{\alpha}\tau_{\Lambda}(\theta)\ket{\alpha}$ and use the Chernoff bound to find the additive approximation to plat closure of $\theta$. The state $\ket{\Phi}$ is the maximally entangled state over the irrep $\Lambda\otimes\Lambda$ embedded in $\D(G)\otimes\D(G)$. This can be prepared by first creating a maximally entangled state over a space of dimension $V_{\Lambda}^{\otimes 2}$. Then we apply the embedding isometry to embed this state into $\C[\D(G)]^{\otimes  2}$.
\begin{theorem}
Given a braid $b\in B_n$ of size $m$, a finite group $G$ and any irrep $([g],\rho)$ of $\D(G)$, there exists a quantum algorithm to additively approximate the trace and plat closures of the braid when the strands are colored by the irrep $([g],\rho)$ of $\D(G)$, in time O(poly($m,n,\log|G|,1/\epsilon$)), where $\epsilon$ is the error in approximation.
\end{theorem}

\subsection{Classical algorithm for fluxon irreps of \texorpdfstring{$\D(G)$}{D(G)}}\label{ClassicalAlg}
In this subsection, we observe that there is a classical randomized algorithm to additively approximate link invariants for certain, combinatorial irreps of $\D(G)$ called ``fluxon" irreps. When the strands are colored by irreps of the type $\Lambda=(h,\text{tr})$, where tr is the trivial irrep of $Z(h)$, then it is possible to approximate the link invariants classically. In these irreps the $R$ matrix has the action of a permutation matrix.
\begin{equation}\label{eq:fluxon-permutation}
R\ket{g_1}\otimes\ket{g_2} = \ket{g_2g_1g_2^{-1}}\otimes\ket{g_2}\,.
\end{equation}
This means that for any braid $\theta$, $\tau_{\Lambda}(\theta)$ is a permutation matrix whose matrix entry can be determined efficiently classically (efficient in the number of crossings in $\theta$). Therefore, we can take elements uniformly at random from the basis vectors of $V_\Lambda^{\otimes n}$ and compute the inner product with $\tau_{\Lambda}(\theta)$. Then, by the Chernoff bound, we obtain an additive approximation to the trace closure. The difference between the classical and quantum algorithms is that, in the quantum case, one has the Hadamard test to efficiently estimate the matrix entries of a unitary operator. Here since the $R$ matrix is a permutation matrix, its entries are easy to determine classically. Therefore, we have
\begin{theorem}
Given a braid $b\in B_n$ of size $m$, a finite group $G$ and any conjugacy class $[g]$ of $G$, there exists a classical algorithm to additively approximate the trace and plat closures of the braid when the strands are colored by the irrep $([g],id)$ of $\D(G)$, in time O(poly($m,n,\log|G|,1/\epsilon$)), where $\epsilon$ is the error in approximation.
\end{theorem}

\paragraph{Remark.} 
While the quantum algorithm above is efficient for all irreps, the classical algorithm is efficient only for invariants coming from ``fluxon" irreps. It may well be the case that there are no classical algorithms for general irreps since this problem is related to approximating elements of irreducible representations of finite groups. For groups like $S_n$, no efficient classical algorithms are known.

\section{Computational complexity of approximating link invariants}\label{complexity}
\newcommand{\Aut}{\operatorname{Aut}}

For this section alone, we assume that the group $G$ is fixed. In this section, we show that providing an additive approximation to the plat closure of a braid (when the strands are colored by an irrep of $\D(G)$) is $\BPP$-hard; likewise, we show that providing a multiplicative approximation is $\SBP$-hard. For the case of fluxons, however, since we have a randomized classical algorithm (presented in Section~\ref{ClassicalAlg}) to additively approximate plat closure, this problem is $\BPP$-complete.\footnote{When we say $\BQP$-complete ($\BPP$-complete, respectively), we mean Promise$\BQP$-complete (Promise$\BPP$-complete, respectively). There are no known problems which are $\BQP$-complete ($\BPP$-complete, respectively).} Our results also imply that computing the plat closure exactly is $\sP$-complete. In the case of link invariants arising from the quantum group $\U_q(\SU(2))$ such as the Jones polynomial, it has been shown that approximating them additively is $\BQP$-complete and approximating them multiplicatively is $\sP$-hard. The complexity of these $\U_q(\SU(2))$ invariants appears to be closely related to the fact that the image of the braid group representations is dense (in the unitary group). Indeed, \citet{Kuperberg:density} shows that denseness is a sufficient condition for $\BQP$-completeness, i.\,e., if the image of the braid group is dense, then this implies the $\BQP$-completeness of additive approximations and $\sP$-completeness of multiplicative approximations. By contrast, as mentioned above, the image of the braid group representations arising from $\D(G)$ is finite \cite{FiniteImage}, which may explain the difference in complexity. A conjecture along these lines is made by Rowell \cite{Rowell}.

To establish our results, we first show that given any randomized computation on $n$ bits, there exists a braid on $\text{poly}(n)$ strands such that the probability of success of the randomized computation is arbitrarily close to the plat closure of the braid. From this, it follows that approximating the plat closure additively is $\BPP$-hard and approximating it multiplicatively is $\SBP$-hard. It also follows that computing it exactly is $\sP$-complete. The notion of additive approximation was defined in Section~\ref{QuantumAlg}; we reproduce it here for convenience along with the definition of multiplicative approximation~\cite{ApproximateCounting}.
\begin{definition}
Given any function $f:D\rightarrow \C$ and a normalization $u:\Z^{+}\rightarrow\R^{+}$, an \textbf{additive approximation} for the pair $(f,u)$ is a probabilistic algorithm which given any $x\in D$ and $\epsilon>0$ produces an output $g(x)$, such that
\[
\Pr[|f(x)-g(x)|>\epsilon u(|x|)]<1/4\,,
\]
in time polynomial in $|x|$ and $\epsilon^{-1}$.
\end{definition}

A multiplicative approximation is a special case of a more general value-dependent approximation defined by \citet{Kuperberg:density}. A probabilistic algorithm to multiplicatively approximate a function is called an $\mathsf{FPRAS}$ (fully polynomial randomized approximation scheme).
\begin{definition}
Given $f:D\rightarrow \C$, a \textbf{multiplicative approximation} or an $\mathsf{FPRAS}$ is a probabilistic algorithm which, given any $x\in D$ and $\epsilon>0$, produces an output $g(x)$ such that
\[
\Pr[|f(x)-g(x)|>\epsilon f(x)]<1/4\,,
\]
in time polynomial in $|x|$ and $\epsilon^{-1}$.
\end{definition}
If the algorithm is deterministic, then we have an $\mathsf{FPTAS}$ (fully polynomial time approximation scheme). Next, we define the class $\SBP$.
\begin{definition}
A language $L$ is in $\SBP$ if there exists a polynomial $p$, a constant $\epsilon>0$, and a polynomial time probabilistic algorithm such that for any input $x$
\[
x\in L \quad\Longrightarrow \quad \Pr[\text{$A(x)$ accepts}]>(1+\epsilon)2^{-p(|x|)}\,,
\]
\[
x\notin L \quad\Longrightarrow\quad \Pr[\text{$A(x)$ accepts}] <(1-\epsilon)2^{-p(|x|)}\,.
\]
\end{definition} 
The class $\SBP$ (small bounded probabilistic $\mathsf{P}$) was defined in \citet{SBP}, which also established several fundamental properties of the class. It was shown that $\SBP$ contains $\mathsf{NP}$ and, in fact, lies between $\mathsf{MA}$ and $\mathsf{AM}$. From the definitions of multiplicative approximation and $\SBP$, we can see that multiplicatively approximating the success probability of an arbitrary randomized computation is (precisely) $\SBP$-hard.

We now state and prove the main theorem.
\begin{theorem}\label{Thm}
Let $\Pl(b)$ be the plat closure of a braid where the strands are colored by any conjugacy class of the alternating group $A_m$, $m\geq 5$, which has at least $4$ fixed points. Then, additively approximating the plat closure is $\BPP$-complete, multiplicative approximation is $\SBP$-complete and exact evaluation is $\sP$-complete.
\end{theorem}

The next lemma is the principal simulation result from which the hardness results will follow.
\begin{lemma}
Given a randomized computation on $n$ bits and any $\epsilon >0$, there exists a braid on poly$(n,{1}/{\epsilon})$ strands and $\text{poly}(n,{1}/{\epsilon})$ crossings with the strands colored by the conjugacy class irrep of the alternating group ($A_m$, $m\geq 6$) with at least $4$ fixed points, such that if the the plat closure of the braid is $\Pl$ and the probability of success of the randomized computation is $\text{P}_s$ then we have
\[
|\text{P}_s - \Pl| < \epsilon \,.
\]
\end{lemma}

\begin{proof}
We use the following model for randomized computation. Let $n$ be the input size. We consider $\text{poly}(n)$ $d$-level systems ($d$its), where $d$ could be exponential in $n$; we label the levels $0$ through $d-1$. The computation begins with $k=\text{poly}(n)$ $d$its in the $0$ state and $\ell=\text{poly}(n)$ random $d$its. Reversible computation is then carried out on this input state using the Toffoli gate. We may organize the computation so that acceptance is indicated by the final value of the first $d$it, adopting the convention that the value $0$ corresponds to acceptance. Following the computation, we involve a untouched ``conclusion'' $d$it in the $0$ state and XOR the first $d$it onto it. This is followed by reversing the computation to recover the initial state (without touching the ``conclusion'' $d$it). With these conventions, the computation accepts precisely when the final state is the same as the initial state. Let the deterministic reversible circuit be represented by the matrix $M$ and the random $d$its be $r_1\dots r_\ell$. Then the probability of success is
\[
\Pr[\text{Accept}]=\frac{1}{d^\ell}\sum_{r_1,\dots, r_\ell}\bra{0^kr_1\dots r_\ell} M \ket{0^kr_1\dots r_\ell}\,.
\]
Since $M$ is deterministic, it is a permutation matrix and we can rewrite the probability of success as
\[
\Pr[\text{Accept}]=\frac{1}{d^\ell}\sum_{\substack{r_1,\dots, r_\ell\\ s_1,\dots, s_\ell}}\bra{0^kr_1\dots r_\ell} M \ket{0^ks_1\dots s_\ell}=\bra{\phi}M\ket{\phi}\,,
\]
where $\ket{\phi}=\frac{1}{\sqrt{d^\ell}}\sum_{r_1\dots r_\ell}\ket{0^kr_1\dots r_\ell}$.

In order to relate the success probability to plat closure of a braid where the strands are colored by fluxon irreps (or conjugacy classes of $G$), we use the Ogburn-Preskill encoding~\cite{OgburnPreskill}. We take as the $d$ level system a pair of strands closed in a plat, so that the state on the pair is $\frac{1}{\sqrt{d}}\sum_{g\in C}\ket{g,g^{-1}}=\frac{1}{\sqrt{d}}\sum_r \ket{r}$ if the conjugacy class $C$ is of size $d$. Therefore, pairs of wires closed in a plat can be used as random $d$its. With this encoding, we have that the probability of success is 
\[
\Pr[\text{Accept}]=\frac{1}{d^\ell}\sum_{\substack{g_1,\dots, g_\ell\\h_1,\dots ,h_\ell}}\bra{(c_1,c_1^{-1})^kg_1g_1^{-1}\dots g_\ell g_\ell^{-1}} M \ket{(c_1,c_1^{-1})^k h_1 h_1^{-1}\dots h_\ell h_\ell^{-1}}\,,
\]
where the $g_i$ and the $h_i$ run over the entire conjugacy class and $\ket{c,c^{-1}}$ represents the $\ket{0}$ state. These can be thought of as variables and the $d$its in the zero state as group constants since they need to be in a specific state. If there were no need for $d$its in the zero state and if $M$ could be realized as a braid, then it is easy to see that the above expression is exactly the plat closure. In order to realize $M$ as a braid, however, we will need to have access to group constants. Therefore, group constants will be necessary for both implementing the circuit and for enforcing an appropriate starting state. Now, if group constants are available, then $M$ can be realized as a braid if the group is sufficiently strong (such as a simple group). It has been shown by \citet{MaurerRhodes} that in simple groups, any function can be realized as a word in the variables and group constants. To construct only the Toffoli gate, less rich groups might work as well (see \cite{Barrington}).

In order to produce group constants with plat closure, our strategy is to introduce more strands and braid them so that when they are closed in a plat, the only solution to the resultant set of equations over the group, is the set of group constants we need. We can then normalize the plat closure so that this part of the braid contributes unity. Assume that the constants needed are $c_1, c_1^{-1}, \dots c_m, c_m^{-1}, (c_1, c_1^{-1})^k$, where $c_1, \dots , c_m$ generate the group (we assume that $G$ has a conjugacy class which contains its own inverses and can generate the group).

The remainder of the proof focuses on the issue of generating group constants. We begin with a brief discussion of the Wirtinger presentation of the knot group and its relationship to the invariants associated with fluxon irreps.

\paragraph{The invariant at a fluxon irrep} 
Recall the Wirtinger presentation of the knot group $\mathcal{K}(K) = \pi_1(\R^3
\setminus K)$ of a knot $K$ (or link): beginning with a knot
diagram\footnote{Recall that a \emph{knot diagram} is a 2-dimensional
  projection of a knot in general position, so that no three lines
  intersect at a point, with explicit annotations that determine, for
  each crossing, which line occludes the other.
  Figure~\ref{fig:stroke} is a knot diagram of the
  unknot. Figure~\ref{fig:crossing} is an example of the typical
  indication of a crossing.} for $K$,
one introduces an orientation for the knot and a generator $x_s$ for
each \emph{stroke} of the diagram (that is, each connected arc of the
diagram, such as the red component of the knot diagram of
Figure~\ref{fig:stroke}). With each crossing of the diagram, such as
the one pictured in Figure~\ref{fig:crossing} involving the labeled
strokes $x_t$, $x_a$, and
$x_b$, one introduces the relation
\begin{equation}\label{eq:Wirtinger-relation}
x_b =  x_t^{-1}\,
x_a \,x_t \,.
\end{equation}
(One can recover this relation by associating each variable $x_s$ with
the path from a fixed base point that travels around the stroke $s$ in
a direction consistent with the right-hand rule and recording the
effect of ``pushing'' a loop underneath the top strand.) The quotient
of the free group generated by the $\{ x_s\}$ by these generators is
isomorphic to the knot group.

As the knot group is an invariant of the knot, one can immediately
obtain other (weaker) invariants as functions of the knot group. We focus
on
\[
h(K,G) = \bigl| H(K,G)\bigr| \quad\text{where}\quad H(K,G) = \{ \phi : \mathcal{K}(K) \rightarrow G \mid
\text{$\phi$ a homomorphism}\}\,,
\]
the number of homomorphisms of $\mathcal{K}(K)$ into the group
$G$. With the Wirtinger presentation described above, the quantity
$h(K,G)$ can be expressed as the number of maps $h: \{ x_s \}
\rightarrow G$ that satisfy the
relations~\eqref{eq:Wirtinger-relation} in the sense that
\[
\phi(x_c) = \phi(x_t)^{-1} \phi(x_c) \phi(x_t)
\]
for each crossing. Observe that any such map necessarily carries the
generators $\{ x_s\}$ into a single conjugacy class of $G$, and we refine the notation above by focusing on those maps associated with a particular conjugacy class $C$ of $G$:
\[
h(K; C, G) = \bigl| H(K; C, G) \bigr| \quad \text{where}\quad H(K; C, G) =  \{ \phi : \mathcal{K}(K) \rightarrow G \mid
\text{$\phi$ a homomorphism, $\forall x_s, \phi(x_s) \in C$}\} \,,
\]
which we abbreviate $h(K,C)$ (and $H(K,C)$) when $G$ can be inferred from context.

As discussed above, a braid $b \in B_{2n}$ induces a link $L(b)$ via
the plat closure, and any (efficiently presented) knot can be given
such a description (as the knot induced from the plat closure of a braid) in a
straightforward fashion. For a group $G$, and the
fluxon representation $\Lambda = (C,1)$ of $D(G)$, this yields the knot
invariant $\Pl_\Lambda(\theta)$ (where $\theta$ is a braid that yields the knot under the plat closure). For this combinatorial case (arising for a fluxon representation), we have
\[
\Pl_\Lambda = h(\mathcal{K}(K); C, G)\,.
\]
With this equality in place, we shall argue about $h(\mathcal{K}(K); C)$ rather than $\Pl_\Lambda$.

\begin{figure}[ht]
\centering
\begin{subfigure}[b]{0.45\textwidth}
  \centering
  \begin{tikzpicture}[scale=.2]
    \draw[tanglea] (-8,-8) -- (8,8);
    \draw[tanglea] (-8,8) -- (8,-8);
    \draw[tangle] (6,-3) node {$x_t$};
    \draw[tangle] (6,3) node {$x_b$};
    \draw[tangle] (-6.5,-3) node {$x_a$};
  \end{tikzpicture}
  \caption{A crossing in a (oriented) knot diagram.}
  \label{fig:crossing}
\end{subfigure}
\begin{subfigure}[b]{0.45\textwidth}
  \centering
  \begin{tikzpicture}[style=tangle,scale=.1]
    \draw[double] (-10,0) -- (10,0); \draw[double] (-10,0) -- (0,17.3); \draw[double] (0,17.3) -- (10,0);
    
    \draw[double,style=tanglered] (10.5,0) .. controls +(20.5,0) and +(10,17.3) .. (0.5,17.8);
    \draw[double] (-10.5,0) .. controls +(-20,0) and +(-10,17.3) .. (0,17.3);
    \draw[double] (10,0) .. controls +(10,-17.3) and (-20,-17.3) .. (-10,0);
  \end{tikzpicture}
  \caption{A stroke in a knot diagram.}
  \label{fig:stroke}
\end{subfigure}

    
    
\caption{Knot diagrams, strokes, and the Wirtinger presentation.}
\end{figure}
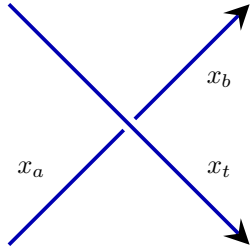
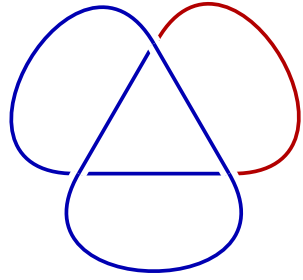

\paragraph{Generating groups constants}
As mentioned above, we assume that $G$ is a finite group and that $C = \{ h^g \mid g \in G\}$ is a conjugacy class that generates $G$. We begin by showing how to construct a knot that---with controlled error---permits us to distinguish a particular set of generators of $G$ (up to automorphism).

Let $\vec{c} = (c_1, \ldots, c_t)$ denote a sequence of elements of $C$ that generate $G$ and
\[
[\vec{c}] = \{ \phi(\vec{c}) = (\phi(c_1), \ldots, \phi(c_k)) \mid \phi \in \Aut(G)\}\,.
\]
Our goal is to construct a knot $k_{\vec{c}}$, together with a distinguished collection of strokes $(s_1, \ldots, s_k)$, with the property that uniform selection of a homomorphism $\psi$ from $H(k,C)$, with overwhelming probability, yields one for which $\psi(s) = (\psi(s_1), \ldots, \psi(s_k)) \in [\vec{c}]$. (Observe that $\Aut(G)$ acts on the collection of legal homomorphisms, so our demand---that the resulting knot distinguish a particular orbit under the $\Aut(G)$ action---is the strongest requirement of this type on which we can insist upon.)

The construction begins with a collection of $k$ (oriented) circles; see Figure~\ref{fig:circles}. Without further constraints, the knot group of this link is free: legal maps $\psi$ may assign arbitrary elements of $C$ to each circle. We introduce constraints among the circles with \emph{band connections}; see Figure~\ref{fig:band-connection}. A \emph{band} is a pair of strands that are constrained to be equal (under any homomorphism into $C$), but are oppositely oriented. They are convenient because of the simple effect they have upon crossing other strokes:
\begin{enumerate}
\item Should a band cross beneath another stroke, as in Figure~\ref{fig:band-under}, the constraints of~\eqref{fig:crossing} preserve the band condition: if the pair of strokes on the left side form a band (are constrained to take equal values under any homomorphism), the same is true of the two strokes on the right. Furthermore, the values assumed by these two bands are related by conjugation of the value of the stroke under which they cross.
\item Should a band cross over another stroke, effectively dividing it into three strokes as in Figure~\ref{fig:band-under}, the constraints of~\eqref{fig:crossing} demand that the ``incoming'' and ``outgoing'' values assumed by the stroke so crossed over are the same, as though the crossing had never happened. The small interior stroke will take on another (conjugate) value, but we shall never allow these to interact with other strokes.
\end{enumerate}

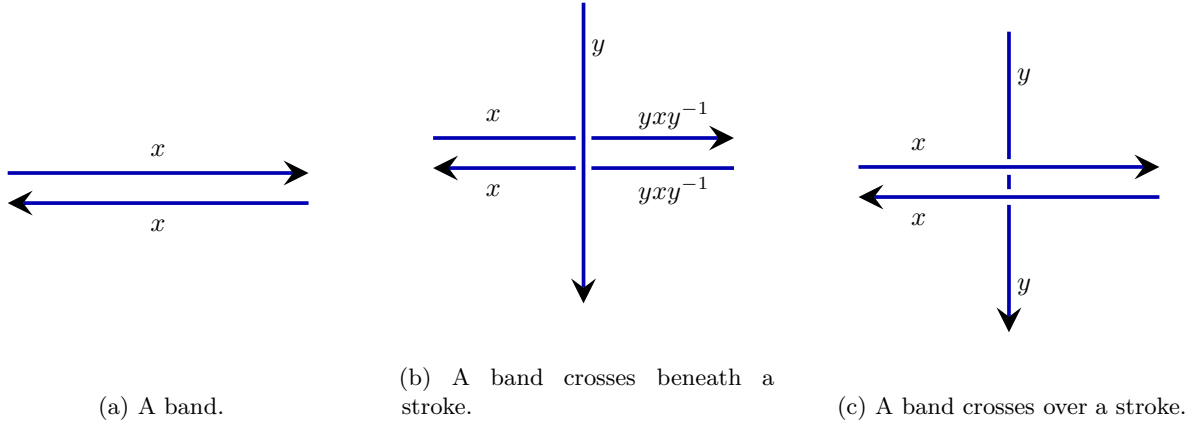
\begin{figure}[ht]
\centering
\begin{subfigure}[b]{0.3\textwidth}
\centering
\begin{tikzpicture}[scale=.2]
\draw[tangle] (0,2.5) node {$x$};
\draw[tangle] (0,-2.5) node {$x$};
\draw[tanglea] (-10,1) -- (10,1);
\draw[tanglea] (10,-1) -- (-10,-1);
\phantom{\draw[double] (0,-10) -- (0,-10);}
\end{tikzpicture}
\vspace{5mm}
\caption{A band.}
\label{fig:band}
\end{subfigure}\qquad
\begin{subfigure}[b]{0.3\textwidth}
\centering
\begin{tikzpicture}[scale=.2]
\draw[tangle] (-6,2.5) node {$x$};
\draw[tangle] (-6,-2.5) node {$x$};
\draw[tangle] (1,7) node {$y$};
\draw[tangle] (6,2.5) node {$y x y^{-1}$};
\draw[tangle] (6,-2.5) node {$y x y^{-1}$};
\draw[tanglea] (-10,1) -- (10,1);
\draw[tanglea] (10,-1) -- (-10,-1);
\draw[tanglea] (0,10) -- (0,-10);
\end{tikzpicture}
\vspace{5mm}
\caption{A band crosses beneath a stroke.}
\label{fig:band-under}
\end{subfigure}\qquad
\begin{subfigure}[b]{0.3\textwidth}
\centering
\begin{tikzpicture}[scale=.2]
\draw[tangle] (-6,2.5) node {$x$};
\draw[tangle] (-6,-2.5) node {$x$};
\draw[tangle] (1,7) node {$y$};
\draw[tangle] (1,-7) node {$y$};
\draw[tanglea] (0,10) -- (0,-10);
\draw[tanglea] (10,-1) -- (-10,-1);
\draw[tanglea] (-10,1)  -- (10,1);
\end{tikzpicture}
\vspace{5mm}
\caption{A band crosses over a stroke.}
\label{fig:band-over}
\end{subfigure}
\caption{Bands and band crossings.}
\end{figure}

Since $c_1, \dots ,c_m$ are conjugates, we can write each $c_i$ as $g^{-1}c_1g$, for a group element $g$. As the $c_i$ generate the group, however, $g$ may be written as a word in the $c_i$ and we have, for each $i$ an equation
\begin{equation}\label{eqn:conjugates}
c_i=c_1^{v_i}\,,
\end{equation}
where $v_i$ is a word in the $c_i$ and their inverses. Introducing a family of indeterminates $x_1, \ldots, x_k$ and replacing each appearance of $c_i$ in~\eqref{eqn:conjugates} with the variable $x_i$, we conclude that the $c_i$ satisfy the equations
\begin{equation}\label{eqn:variable-conjugates}
x_i=x_1^{w_i}\,,
\end{equation}
where $w_i$ is the word in the $x_i$ and their inverses obtained from $v_i$. We now adopt a construction of \citet{Homomorphs}, who shows how to construct a knot that imposes these relations. Each such equation $x_i = x_1^{w_i}$ is imposed by introducing a band between the circle associated with $x_1$ and the circle associated with $x_i$ (see Figure~\ref{fig:band-connection}) and passing it through the circles associated with the variables appearing in $w_i$; by passing such a band ``through'' such a circle---one over and once under as in Figure~\ref{fig:relation}---the generator associated with the circle operates on the band value by conjugation while leaving the generator associated with the circle unaffected. In general, the words $w_i$ may contain references to the variable $x_i$, in which case one passes the band through the circle associated with $x_i$, as in Figure~\ref{fig:self-equation}.

This construction produces a knot we call $K_c^1$. By construction, associating the strokes the $k$ original circles with the variables $x_i$, there is a homomorphism $\psi$ of $\mathcal{K}(K_c^1)$ into $C \subset G$ for which $\psi: x_i \mapsto c_i$. As remarked above, if $\phi \in \Aut(G)$ which fixes the conjugacy class $C$, then $\phi \circ\psi$ is also an element of $H(K_c^1, C)$. However, there will, in general, be other solutions to the equations~\eqref{eqn:variable-conjugates}: in particular, the map which carries all strokes of the knot to a particular, fixed element $c \in C$ satisfies the equations~\eqref{eqn:variable-conjugates}. Our goal is to add some strands to the knot so as to ``amplify'' the solution of interest---the map $\psi: x_s \mapsto c_i$---so that it appears with overwhelming majority if an element of selected at random from $H(\cdot, C)$.


\begin{figure}[ht]\centering
\begin{subfigure}[b]{0.45\textwidth}
\begin{tikzpicture}[scale=.15,radius=5]
\draw[tanglea] (-25,0) arc[start angle=0, end angle=360];
\draw[tangle] (-30,0) node {\textsl{$x_1$}};
\draw[tanglea] (-10,0) arc[start angle=0, end angle=360];
\draw[tangle] (-15,0) node {\textsl{$x_2$}};
\draw[tangle] (-4,0) node {\textsl{$\cdots$}};
\draw[tanglea] (10,0) arc[start angle=0, end angle=360];
\draw[tangle] (5,0) node {\textsl{$x_k$}};
\end{tikzpicture}
\vspace{8mm}
\caption{Initial circles.}
\label{fig:circles}
\end{subfigure}
\qquad
\begin{subfigure}[b]{0.45\textwidth}
\centering
\begin{tikzpicture}[scale=.15,radius=5,style=tangle]
\draw[tanglea] (-10,0) arc[start angle=0, end angle=270];
\draw[double] (15,5) arc[start angle=90, end angle=-180];

\draw[double] (-15, -5) .. controls +(5,0) and +(-5,-9) .. (0,-1);
\draw[double] (0, -1) .. controls +(5,9) and +(0,6) .. (10,0);

\draw[double] (-10,0) .. controls +(0,-6) and +(-5,-9) .. (0,1);
\draw[double] (0,1) .. controls +(5,9) and +(-5,0) .. (15,5);
\end{tikzpicture}
\caption{A band between two circles.}
\label{fig:band-connection}
\end{subfigure}

\vspace{1cm}
\begin{subfigure}[b]{0.45\textwidth}
\begin{tikzpicture}[scale=.15,radius=5,style=tangle]
\draw[tanglea] (-15,0) arc[start angle=0, end angle=270];
\draw (-20,0) node {\textsl{$x$}};

\draw[double] (-0,0) arc[start angle=0, end angle=270];
\draw (-5,0) node {\textsl{$y$}};

\draw[double] (-20, -5) .. controls (-15,-5) and (-10,-11) .. (10,-11);
\draw[double] (-15,0) .. controls (-15,-5) and (-10,-10) .. (10,-10);

\draw[double] (0,0) .. controls (0,-5) and (5,-2) .. (10,-2);
\draw[double] (-5,-5) .. controls (-3,-5) and (8,-3) .. (10,-3);

\draw[tanglea] (15,0) arc[start angle=0, end angle=360];
\draw (10,0) node {\textsl{$z$}};

\draw[double] (10,-2) .. controls (16.25,-2) and (16.25,-11) .. (10,-11);
\draw[double] (10,-3) .. controls (15,-3) and (15,-10) .. (10,-10); 
\end{tikzpicture}
\caption{A simple relation: $x^z = y$.}
\label{fig:relation}
\end{subfigure}
\qquad
\begin{subfigure}[b]{0.45\textwidth}
\begin{tikzpicture}[scale=.15,radius=5,style=tangle]
\draw[tanglea] (-15,0) arc[start angle=0, end angle=270];
\draw (-20,0) node {\textsl{$x$}};

\draw[double] (-0,0) arc[start angle=0, end angle=270];
\draw (-5,0) node {\textsl{$y$}};

\draw[double] (-20, -5) .. controls (-15,-5) and (-10,-3) .. (-5,-3);
\draw[double] (-15,0) .. controls (-15,-5) and (-10,-2) .. (-5,-2);


\draw[double] (-5, -3) .. controls (.5,-3) and (-2,-11) .. (10,-11);
\draw[double] (-5,-2) .. controls (2,-2) and (-2,-10) .. (10,-10);

\draw[double] (0,0) .. controls (0,-5) and (5,-2) .. (10,-2);
\draw[double] (-5,-5) .. controls (-3,-5) and (8,-3) .. (10,-3);

\draw[tanglea] (15,0) arc[start angle=0, end angle=360];
\draw (10,0) node {\textsl{$z$}};

\draw[double] (10,-2) .. controls (16.25,-2) and (16.25,-11) .. (10,-11);
\draw[double] (10,-3) .. controls (15,-3) and (15,-10) .. (10,-10); 
\end{tikzpicture}
\caption{The equation $x^{zy} = y$.}
\label{fig:self-equation}
\end{subfigure}
\caption{Building a set of generators.}
\end{figure}
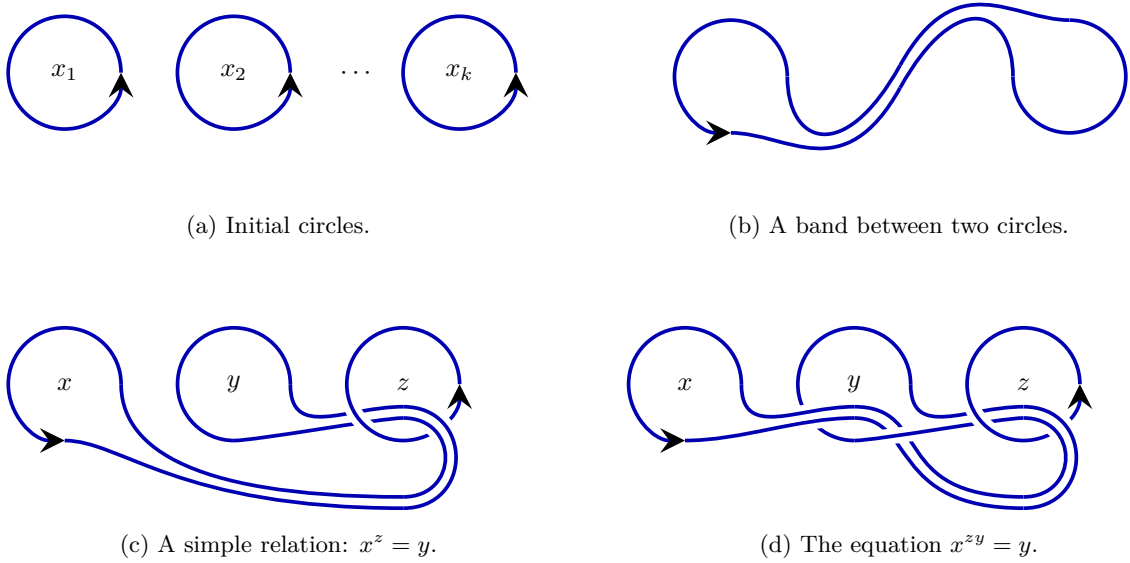

\newcommand{\new}{\operatorname{new}}

To minimize the contribution from solutions not related by an automorphism, consider a undesirable solution $\vec{d} = (d_1, \ldots, d_k)$ (for which there is an homomorphism $\psi': x_i \mapsto d_i$ and yet $d \not\in [c]$). We now introduce a new circle, associated with the variable $y_{\new}$, and as before introduce an equation of the form
\begin{equation}\label{eq:eq-suppress}
y_{\new}=x_1^{w(\vec{x}) y_{\new}}\,,
\end{equation}
where $w(\vec{x})$ is a word in the $x_i$. In order to ``suppress'' the homomorphism associated with $\vec{d}$, we shall choose $w$ in such a way that
\begin{enumerate}
\item $w(\vec{d}) = 1$: thus any homomorphism $\psi'$ which carries $x_i$ to $d_i$ must necessarily carry $y_{\new}$ to the same value as $x_1 = d_1$. (Specifically, the only solution to the equation $y_{\new} = x_1^{w(d)y_{\new}} = x_1^{y_{\new}}$ is $y_{\new} = x_1$.)
\item\label{item:multiple} $w(\vec{c}) = \alpha$, a group constant for which the equation $y_{\new} = x_1^{\alpha y_{\new}}$ has multiple solutions. In this case, the homomorphism $\psi: x_i \mapsto c_i$ may be extended in multiple ways with an assignment to $y_{\new}$ (any solution to $y_{\new} = x_1^{\alpha y_{\new}}$ will suffice).
\end{enumerate}
Lemma~\ref{lem:alternating} below guarantees that there is a nontrivial element $\alpha \in A_n, n \geq 6$ satisfying item~\ref{item:multiple} above. Lemma~\ref{lem:solutions} below guarantees that when $G$ is simple (i.e., has no nontrivial normal subgroups), there is a word $w$ in the variables $x_1, \ldots, x_k$ so that $w(d) = 1$ and $w(c) = \alpha$ for any fixed $\alpha \in G$. Focusing now on a specific group $G = A_n, n \geq 6$, let $\alpha$ be the element which maximizes the number of solutions to the equation $y = c_1^{\alpha y}$. It follows immediately that, for the new knot $K'$ induced from this process,
\[
|\{ \phi \in H(\mathcal{K}(K'); C, G)  \mid \psi(x_i) = c_i \}| \geq 2\,|\{ \phi \in H(\mathcal{K}(K'); C, G) \mid \psi(x_i) = d_i \}|\,.
\]
Repeating this process $\ell$ times for each of the offending vectors $\vec{d}$ induces a knot $K_c^\ell$ so that for any $\vec{d} \not\in [\vec{c}]$ we have
\[
|\{ \phi  \in H(\mathcal{K}(K^\ell_c); C, G) \mid \psi(x_i) = c_i \}| \geq 2^\ell\,|\{ \phi  \in H(\mathcal{K}(K^\ell_c); C, G) \rightarrow C \mid \psi(x_i) = d_i \}|\,.
\]
(Observe that, as presented, this is only efficient if $G = A_n$ has constant size.)

\begin{lemma}\label{lem:alternating}
Consider the alternating group on $m$ symbols $A_m$ and consider any conjugacy class with at least $4$ fixed points, say, $m-3$, $m-2$, $m-1$ and $m$. Since the conjugacy class is non-trivial, it has a $k$ cycle with $k\geq 2$. Pick $\alpha$ in this class to have $(1\,2\,\dots\,k)$ in that $k$ cycle. Now, consider the equation $y=\alpha^{(1\,k)(m-2\,m-3)y}$. This equation has at least two solutions.
\end{lemma}
\begin{proof}
It can easily be checked that two solutions are given by elements that differ from $\alpha$ only in the $k$ cycle. The first one contains $(1\,2\,\dots\,k-1\,m)$ and the second $(1\,2\,\dots\,k-1\,m-1)$ in that $k$ cycle.
\end{proof}


\begin{lemma}\label{lem:solutions}
  Let $G$ be a finite group and $F_k$ be the free group on $k$ generators $\{x_1, \ldots, x_k\}$. For an element $\vec{d} = (d_1, \ldots, d_k) \in G^k$, let $\phi_\vec{d}$ be the evaluation homomorphism that carries $x_i$ to $d_i$ and define
  \[
  A_\vec{d} = \ker \phi_\vec{d} = \{ w \in F_k \mid \phi_\vec{d}(w) = 1 \}\,.
  \]
  For two elements $\vec{c}, \vec{d} \in G^k$, define $C$ to be the subgroup generated by the $c_i$, $D$ to be the subgroup generated by the $d_i$, and 
  \[
  A_\vec{d}(\vec{c}) = \{ \phi_\vec{c}(w) \mid w \in A_{\vec{d}}\}\,.
  \]
  Then
  \begin{enumerate}
  \item \label{part:normal} $A_\vec{d}(\vec{c})$ is normal in $C$. In particular, when the $c_i$ generate $G$, $A_\vec{d}(\vec{c})$ is normal in $G$.
  \item \label{part:nontrivial} If $\vec{c}$ and $\vec{d}$ are not related by a homomorphism of $G$ (which is to say that no homomorphism from $D$ to $C$ carries each $c_i$ to  $d_i$) then $A_\vec{d}(\vec{c}) \neq 1$.
  \end{enumerate}
\end{lemma}

\begin{proof}
For part~\ref{part:normal}, $\ker(\phi_\vec{d})$ is certainly normal in $F_k$ since $\phi_\vec{d}$ is a group homomorphism. Let $C$ be the subgroup of $G$ generated by the $c_i$, and $g \in C$; we wish to show that $g^{-1}A_\vec{d}(\vec{c})g = A_\vec{d}(\vec{c})$. If $w \in F_k$ is a word in the free group for which $\phi_\vec{c}(w) = g$, we have $w^{-1} A_\vec{d} w = A_\vec{d}$ and hence that 
\[
g^{-1} A_{\vec{d}}(\vec{c}) g = \phi_{\vec{c}}(w^{-1} A_\vec{d} w) = \phi_{\vec{c}}(A_\vec{d}) = A_{\vec{d}}(\vec{c})\,,
\]
as desired.

As for part~\ref{part:nontrivial}, let $C$ denote the subgroup generated by $\{ c_i\}$ and $D$ the subgroup generated by $\{ d_i\}$. Observe that if $A_\vec{d}(\vec{c})$ is trivial we have $\ker \phi_\vec{d} \subset \ker \phi_\vec{c}$. In this case, the natural quotient map 
\[
q_{\vec{c}}: D \cong F/\ker \phi_\vec{d} \rightarrow C \cong F/\ker \phi_\vec{c}
\]
yields a homomorphism $\psi: D \rightarrow C$ for which $\psi: d_i \mapsto c_i$. To be precise, let $q_\vec{c}$ denote the quotient map $q_{\vec{c}}: F/\ker \phi_\vec{d} \rightarrow F/\ker \phi_\vec{c}$ and let $i_\vec{d}$ denote the inverse of the isomorphism $\phi_\vec{d}$ induces from $F/\ker \phi_{\vec{d}}$ to $D$; observe that $i_\vec{d}(d_i) = x_i \;(\ker \phi_\vec{d})$. Then the map $\psi = \phi_\vec{c} \circ q_{\vec{c}} \circ i_{\vec{d}}$ is a homomorphism of $D$ onto $C$ that carries each $d_i$ onto $c_i$; see Figure~\ref{fig:homomorphism}.

\begin{figure}[ht]
\begin{center}
\begin{tikzpicture}
    \node (D) at (0,0) {$D$};
    \node (quotient-D) at (3,0) {$F/\ker \phi_\vec{d}$};
    \node (quotient-C) at (6,0) {$F/\ker \phi_\vec{c}$};
    \node (C) at (9,0) {$C$};
    \path[->] (D) edge [bend right] node [below] {$i_{\vec{d}}$} (quotient-D);
    \path[->] (quotient-D) edge [bend right] node [above] {$\phi_\vec{d}$} (D);
    \path[->] (quotient-D) edge node [above] {$q_\vec{c}$}  (quotient-C);
    \path[->] (quotient-C) edge node [above] {$\phi_\vec{c}$} (C);
\end{tikzpicture}
\end{center}
\caption{The homomorphism $\psi: d_i \mapsto c_i$.}
\label{fig:homomorphism}
\end{figure}

    
  \end{proof}
  The above lemma implies that if $G$ is simple then $H_1=G$ ($H_1$ cannot be trivial because $\vec{c}$ and $\vec{d}$ are not related by an automorphism). This means that there exists a word $w$ in variables $x_i$ such that $w(\vec{c})=z$ and $w(\vec{d})=1$.
\end{proof}

\begin{proof}[Proof of Theorem~\ref{Thm}]
The first two follow from the definitions of $\BPP$ and $\SBP$ and the following lemmas. The $\sP$-completeness follows from the fact that exact evaluation of the success probability of a randomized computation is $\sP$-complete.
\end{proof}

\section{Quantum computation with anyons}\label{QC:anyons}
Anyons are particles which exist in two dimensions and have exotic statistics. Anyons are useful for quantum computation because quantum information can be stored on a system of anyons in a non-local fashion. This means that local errors do not corrupt the quantum information and a computer based on anyons will be inherently fault tolerant to local errors. For a tutorial on quantum computation using anyons, see the notes by Preskill~\cite{PreskillNotes}. In this section, we give short introduction to anyons described by $\D(G)$ for a finite group $G$. Then we show how one can use the QFT over $\D(G)$ to simulate an anyon computer efficiently. This has been shown in \cite{PreskillNotes, Mochon1}, but it was assumed that the dimension of the irreps of $\D(G)$ are of constant size. Here we use the QFT over $\D(G)$ to simulate an anyon computer in potentially large irreps of $\D(G)$.

The Hilbert space of an anyon (whose symmetries are described by $\D(G)$) is an irreducible representation of $\D(G)$. Recall that irreps of $\D(G)$ are characterized as $(h,\rho)$, where $h$ is a representative element of a conjugacy class and $\rho$ is an irrep of $Z(h)$. The group elements $h$ are called fluxes and the irreps $\rho$ are called charges. An anyon, in general, has both flux and charge. Anyons which transform as $(h,\text{tr})$, where tr is the trivial irrep of $Z(h)$, are called \emph{fluxons}. On the other hand, if we pick $h$ to be the identity element of the group, then $Z(h)=Z(e)=G$. Anyons described by the irreps $(e,\rho)$, where $\rho$ is an irrep of $G$, are called \emph{chargeons}. Recall from \eqref{R:InsideIrrep} that the action of the $R$ matrix on a pair of anyons is
\begin{equation}R\ket{g_1,v_1}\otimes\ket{g_2,v_2} = \ket{g_2g_1g_2^{-1}, \rho(k_{g_2g_1g_2^{-1}}^{-1}g_2k_{g_1})v_1}\otimes\ket{g_2,v_2}.
\end{equation}
In the special case of fluxons, this action reduces to
\begin{equation}R\ket{g_1}\otimes\ket{g_2} = \ket{g_2g_1g_2^{-1}}\otimes\ket{g_2} .
\end{equation}
This is the action when we wind the first anyon around the second in the anticlockwise direction. The braid operator $TR$ is given by
\begin{equation}B\ket{g_1}\otimes\ket{g_2} = \ket{g_2}\otimes\ket{g_2g_1g_2^{-1}} .
\end{equation}
The action on chargeons can be determined similarly. Given an irrep $(g,\rho)$, the conjugate irrep is $(g^{-1},\bar{\rho})$, where $\bar{\rho}$ is the conjugate irrep of $\rho$. The action of $\bar{\rho}$ is simply the complex conjugation of the action of $\rho$. Given two anyons in conjugate irreps, the state of trivial total flux and charge is the maximally entangled state
\[
\ket{\Phi} = \frac{1}{\sqrt{|C|d_\rho}}\sum_{g,v} \ket{g,v}\otimes\ket{g^{-1},v^\ast},
\]
where the sum is over all $g$ in the conjugacy class $C$ and all $v$ in the irrep $\rho$.

\subsection*{Simulation of anyons}

As described in~\cite{PreskillNotes}, in order to perform universal quantum computation, we need to be able to
\begin{enumerate}
\item \textit{Prepare any state in the Hilbert space of a pair of anyons which correspond to conjugate irreps.}
\item \textit{Perform braiding of anyons around each other and around ancillas.}
\item \textit{Fuse pairs of anyons and measure the flux and charge of the resulting particle.}
\end{enumerate}
It can be seen easily that in order to simulate each of these steps, one needs the Fourier transform and the Clebsch-Gordan transform over $\D(G)$. Initial state preparation can be carried out if we can construct states which lie inside pairs of irreps of $\D(G)$. To do this, we can use the Clebsch-Gordan transform. First we embed the state in the direct sum of the CG decomposition and then perform the inverse CG transform. Next, in order to perform braiding on the states, we need to implement the $R$ matrix inside the irrep. We use the same trick as before and implement the $R$ matrix in the regular representation of $\D(G)$ and use the QFT\@.  Finally, in order to simulate fusion, we again make use of the CG transform since fusion of anyons is a CG transform followed by a measurement in the computational basis. Since we know how to perform a Fourier transform over some groups, we now focus on the CG transform.

\section{Clebsch-Gordan decomposition}\label{sec:CG}
In this section, we first describe an efficient algorithm to perform the Clebsch-Gordan decomposition over $\D(G)$ for fluxon irreps. Then we give an efficient algorithm for general irreps under certain conditions. The following theorem from Curtis and Reiner \cite{CurtisAndReiner} is useful to understand the Clebsch-Gordan transform.
\begin{theorem}
Suppose we have two subgroups $H$ and $K$ of $G$ and two representations $\rho$ and $\sigma$ of $H$ and $K$ respectively. Suppose that we induce these two representations to $G$, then the tensor product of the two $G$ representations can be decomposed into a direct sum of induced representations as
\[
\rho\uparrow^G\otimes \sigma\uparrow^G \cong \bigoplus_d \left(\rho\downarrow_{H\cap K^d} \otimes \sigma\downarrow_{H\cap K^d}\right)\uparrow^{G} \,,
\]
where $d$ runs over all $(H,K)$ double coset representatives.
\end{theorem}
We can apply this theorem to irreps of $\D(G)$ since they are all induced representations from centralizer subgroups. Suppose that we have the tensor product of two irreps of $\D(G)$, say $([g],\rho)$ and $([h],\sigma)$, then the theorem implies that we can write this as
\begin{equation}\label{CG1}
\rho\uparrow^G\otimes \sigma\uparrow^G \cong \bigoplus_d \left(\rho\downarrow_{Z(g)\cap Z(h)^d} \otimes \sigma\downarrow_{Z(g)\cap Z(h)^d}\right)\uparrow^{G} \,.
\end{equation}
However, this is still not in the form that we want. In order to obtain it, we need to understand the double coset representatives. We show that the $(Z(g),Z(h))$ double coset representatives also label the different conjugacy classes that appear in the product $[g]\cdot [h]$ in the conjugacy class algebra. To see this let the conjugacy class of $g$ be $\{g, k_1gk_1^{-1},\dots , k_n g k_n^{-1}\}$ and that of $h$ be $\{h, l_1 h l_1^{-1}, \dots , l_m h l_m^{-1}\}$, where $k_i$ and $l_i$ label the complete set of coset representatives of $Z(g)$ and $Z(h)$ respectively. Now consider all possible products of the elements of the two sets. In order to determine the different conjugacy classes that appear in the products, we only need to consider elements of the form $g l_i h l_i^{-1}$ since anything of the form $k_j g k_j^{-1} l_i h l_i^{-1}$ can be conjugated by $k_j^{-1}$ to get the former type. 

This is still not enough since two elements $g l_i h l_i^{-1}$ and $g l_j h l_j^{-1}$ could be conjugates. In that case, there must be an element $z$ of $Z(g)$ such that $z (g l_i h l_i^{-1}) z^{-1}=g (z l_i) h (z l_i)^{-1}=g l_j h l_j^{-1}$. This means that $l_j$ is in the same double coset of $(Z(g),Z(h))$ as $l_i$. In other words, the right action of $Z(g)$ on left cosets of $Z(h)$ determines the different conjugacy classes which is exactly the different $(Z(g),Z(h))$ double coset representatives. Picking a set of double coset representatives $d$, we can say that the different conjugacy classes that appear in the product $[g]\cdot [h]$ are $[g h^d]$ for all the different double coset representatives $d$.

We now need to determine the number of times each $[g h^d]$ appears in $[g]\cdot [h]$ since this determines (some of the) multiplicities in the Clebsch-Gordan decomposition. It is enough to determine the number of times $[gh]$ appears in the product $[g]\cdot [h]$ since its entire conjugacy class would appear the same number of times (for any other $d$ the procedure is the same). In order to count this, consider the two groups $Z(gh)$ and $Z(g)\cap Z(h)$. It is easy to see that the latter is a subgroup of the former. If it is a strict subgroup, then there exist elements in $Z(gh)$ which do not commute with either $g$ or $h$ or both. It turns out that in fact, any element of $Z(gh)$ either commutes with both $g$ and $h$ or does not commute with both. To see this, notice that if $s$ is an element not in $Z(g)\cap Z(h)$ and it commutes with $g$ but not $h$, then we have that $gh=(gh)^s=g^sh^s=gh^s$. But this implies that $h^s=h$ and so $s$ commutes with $h$. This means that any element in $Z(gh)$ that does not commute with $g$, also does not commute with $h$ and vice versa. 

Now for any non-trivial coset representative $s$ of $Z(g)\cap Z(h)$ in $Z(gh)$, we have that $gh=s gh s^{-1} = g^s h^s$. Since $s$ does not commute with $g$ and $h$, we have produced a pair $(g^s, h^s)$ distinct from $(g,h)$ such that their product is the same. For distinct coset representatives $s$ and $t$, the pairs $(g^s,h^s)$ and $(g^t,h^t)$ are distinct since if not, then $g^s=g^t$ and $h^s=h^t$. This means that $st^{-1}\in Z(g)\cap Z(h)$ which is not possible since $s$ and $t$ are in distinct cosets of $Z(g)\cap Z(h)$. Therefore, for each distinct coset representative $s$ of $Z(g)\cap Z(h)$ in $Z(gh)$, we get a distinct pair $(g^s,h^s)$ such that $g^s h^s= g h$. The same argument holds for any non-trivial $d$. To show that these are all the possible pairs, we make use of the above theorem. When $\rho$ and $\sigma$ are both trivial, equation~\eqref{CG1} takes the form
\[
([g],\text{tr})\otimes ([h],\text{tr}) \cong \bigoplus_d \left(\text{tr}\uparrow_{Z(g)\cap Z(h^d)}^{Z(gh^d)}\right) \uparrow^G\,.
\]
This shows that the number of times $[gh^d]$ appears in $[g]\cdot [h^d]$ is the index of $Z(g)\cap Z(h^d)$ in $Z(gh^d)$. 

With all this in hand, we can determine the Clebsch-Gordan decomposition. We can re-write equation~\eqref{CG1} to get
\[
\rho\uparrow^G\otimes \sigma\uparrow^G \cong \bigoplus_d \left(\left(\rho\downarrow_{Z(g)\cap Z(h)^d} \otimes \sigma\downarrow_{Z(g)\cap Z(h)^d}\right)\uparrow^{Z(gh^d)}\right)\uparrow^G \,,
\]
by taking the induction in two stages. Now suppose that $\left(\rho\downarrow_{Z(g)\cap Z(h)^d} \otimes \sigma\downarrow_{Z(g)\cap Z(h)^d}\right)\uparrow^{Z(gh^d)}$ breaks up into irreps of $Z(gh^d)$ as
\[
\left(\rho\downarrow_{Z(g)\cap Z(h)^d} \otimes \sigma\downarrow_{Z(g)\cap Z(h)^d}\right)\uparrow^{Z(gh^d)} \cong \bigoplus_{\mu\in \widehat{Z(gh^d)}} m_{\rho,\sigma,\mu}\mu\,,
\] 
where $m_{\rho,\sigma,\mu}$ is the multiplicity of the irrep $\mu$. We now obtain the Clebsch-Gordan decomposition as 
\begin{equation}\label{CG2}
\rho\uparrow^G\otimes \sigma\uparrow^G \cong \bigoplus_d \bigoplus_{\mu\in \widehat{Z(gh^d)}} m_{\rho,\sigma,\mu}(\mu\uparrow^{G}) \,.
\end{equation}
Since this decomposition is obtained by considering the action of $G$ alone, we need to check if the action of $h^\ast$ is consistent with it. We do this in the next section as we develop the transform.

\subsection{Clebsch-Gordan transform}
\subsubsection{Fluxon irreps}
We first describe the Clebsch-Gordan transform for irreps of $\D(G)$ of the type $([g],\text{tr})$ (fluxon irreps). For this case, we give an efficient transform for any finite group. For two fluxon irreps, the Clebsch-Gordan decomposition states that
\[
([h],\text{tr})\otimes ([g],\text{tr}) \cong \bigoplus_d \left(\text{tr}\uparrow_{Z(g)\cap Z(h^d)}^{Z(gh^d)}\right) \uparrow^G\cong \bigoplus_d\bigoplus_{\mu\in\widehat{Z(gh^d)}}m_{\text{tr},\text{tr},\mu}(\mu\uparrow^G)\,.
\]

We give the transform in two steps, one for each of the above two isomorphisms. In the first step, we need to convert from the basis $\ket{h^a,g^b}$ to $\ket{(gh^d)^k,t}$, where $t$ is an element of the transversal $T_d$ of $Z(g)\cap Z(h^d)$ in $Z(gh^d)$. Once we fix the double coset representatives $d$ and the transversal $t$ (for each $d$), this transformation is straightforward. Given $\ket{h^\prime,g^\prime}$, we determine their product and the double coset representative $d$ that this product belongs to. Recall that for each element $t$ in $T_d$, we have that $g^t (h^d)^t= g h ^d$ and that each such $t$ gives a distinct pair $(g^t, (h^d)^t)$. Therefore, we can determine $t$ after ordering the pairs. This gives us the first transformation. It is also clear from this that the co-algebra action is consistent across the both sides. Indeed, the action of any $c^\ast$ on a state of the form $\ket{t_1,v_1}\otimes\ket{t_2,v_2}$ is
\[
c^\ast (\ket{t_1,v_1}\otimes\ket{t_2,v_2})=\sum_{h_2h_1=c}h_1^\ast\ket{t_1,v_1}\otimes h_2^\ast\ket{t_2,v_2} = \delta_{\{t_2t_1=c\}}\ket{t_1,v_1}\otimes\ket{t_2,v_2} \,.
\]
Under the Clebsch-Gordan transform, this state is taken to $\ket{t_2t_1,v_3}$, where $v_3$ is a state determined by $v_1$ and $v_2$. This state is given by the second transformation described below. Therefore, the action of $c^\ast$ is 
\[
c^\ast\ket{t_2t_1,v_3}=\delta_{\{t_2t_1=c\}}\ket{t_2t_1,v_3}\,.
\]
This shows that the coalgebra action is consistent with this decomposition.

For the second transformation, we need to block diagonalize the induced representation of the trivial irrep of $Z(g)\cap Z(h^d)$ to $Z(gh^d)$. We show how to do this if we can perform efficient QFT over both groups. Suppose that $\rho$ is an irrep of a group $B$ which is a subgroup of the group $A$ and suppose that we can perform a QFT over $A$ and $B$, then we show how to use these QFTs to give an efficient way to block diagonalize the induced representation $\rho \uparrow^A$. We first embed the induced representation into $\C[A]$ in the following way. The induced representation consists of vectors of the form $\ket{t,v}$, where $t$ is an element of the transversal of $B$ in $A$ and $v$ is a vector in the irrep space $\rho$. We embed $\ket{v}$ into $\C[B]$ by taking an ancilla of size $|B|/d_\rho$, where $d_\rho$ is the dimension of $\rho$. This is possible if we know how to perform the QFT over $\C[B]$. Now that we have an embedding into $\C[A]$, we can perform the QFT over $A$ and obtain a basis $\ket{\rho,i,j}$ where only some irreps $\rho$ appear. Since we know which ones do not appear at all, we can label the irreps of $A$ such that the state is $\ket{0}\otimes \ket{\rho,i,j}$ where the first register is of size $|H|/d_\rho$. We can now discard this register and we obtain the required block diagonalization. Using this technique when $\rho$ is a trivial representation of $B=Z(g)\cap Z(h^d)$ and $A=Z(gh^d)$ gives us a way to decompose the induced representation from the trivial representation of a subgroup. This gives us the second transformation and completes the Clebsch-Gordan transform for fluxon irreps. We have thus shown that if we can perform an efficient QFT over $Z(g)$ and over $Z(g)\cap Z(h)$ for all $g,h\in G$, then we can perform an efficient Clebsch-Gordan transform over $\D(G)$.

\subsubsection{General irreps}
For the more general irreps, we show that one can perform efficient Clebsch-Gordan transform if we can
\begin{enumerate}
\item perform QFT and CG transforms over $Z(h)$ and $Z(g)\cap Z(h)$ for all $g,h\in G$ and,
\item block diagonalize irreps of centralizers restricted to intersections of centralizers.
\end{enumerate}
The procedure can be split into three steps in the following way.
\begin{align}
([h],\rho)\otimes ([g],\sigma) &\cong \bigoplus_d \left(\left(\rho\downarrow_{Z(g)\cap Z(h^d)}\otimes\sigma\downarrow_{Z(g)\cap Z(h^d)} \right)\uparrow^{Z(gh^d)}\right) \uparrow^G \nonumber \\
&\cong \bigoplus_d \left(\left(\bigoplus_{\nu\in \widehat{Z(g)\cap Z(h^d)}} n_{\nu}\nu \right)\uparrow^{Z(gh^d)}\right) \uparrow^G \nonumber \\
&\cong \bigoplus_d\bigoplus_{\mu\in\widehat{Z(gh^d)}}m_{\mu}(\mu\uparrow^G)\,.
\end{align}
In the above, $n_\nu$ is the multiplicity of $\nu$ (an irrep of $Z(g)\cap Z(h)$) in the tensor product decomposition of $\rho$ and $\sigma$ restricted to $Z(g)\cap Z(h^d)$ and $m_\mu$ is the multiplicity of $\mu$ (an irrep of $Z(gh^d)$) when $\nu$ is induced to $Z(gh^d)$ and decomposed into irreps of $Z(gh^d)$. These multiplicities depend on $\rho$, $\sigma$ and $d$ in general.
\begin{enumerate}
\item The first step takes us from the basis $\ket{h^a,v_1,g^b,v_2}$ to the basis $\ket{(gh^d)^k, t, v_1,v_2}$, where $v_1$ and $v_2$ are vectors in the irrep spaces of $\rho$ and $\sigma$. This step is the same as in the previous case (for fluxon irreps) since we do not operate on the vectors $v_1$ and $v_2$. The transversal is picked in the way described above. 
\item The second step can be done if we know how to decompose any irrep of $Z(g)$ and $Z(h^d)$ into irreps of $Z(g)\cap Z(h)^d$ and then perform Clebsch-Gordan transform over the group $Z(g)\cap Z(h)^d$. This step may be done for particular groups.
\item The third step can be done using the procedure described above for decomposing induced irreps, since $\nu$ is an irrep of $Z(g)\cap Z(h^d)$ (the group $B$ above) and it is induced to $Z(gh^d)$ (the group $A$ above). 
\end{enumerate}
This gives an efficient algorithm for the CG transform under certain conditions. Next, we show how these conditions are satisfied for $\D(\mathbb{Z}_p\rtimes\mathbb{Z}_q)$ which is sufficient for universal quantum computation \cite{Mochon2}.

\subsection{Fourier and Clebsch-Gordan transforms over \texorpdfstring{$\D(\mathbb{Z}_p\rtimes\mathbb{Z}_q)$}{metabelian groups}}
It is shown in \cite{Mochon2}, that the group $G=\mathbb{Z}_p\rtimes\mathbb{Z}_q$ can be used to perform universal quantum computation when $p$ and $q$ are prime and $q|(p-1)$. Here we show how to perform the Fourier and Clebsch-Gordan transforms over $\D(G)$. First, in order to fully describe the group, we have to pick the homomorphism from $\mathbb{Z}_q$ to $\text{Aut}(\mathbb{Z}_p)$. Such homomorphisms are characterized by elements $\alpha$ such that $\alpha^q=0$ mod $p$. Having picked such an $\alpha$, we see that the group multiplication in $G$ is $(a_1,b_1)(a_2,b_2)=(a_1+a_2 \alpha^{b_1},b_1+b_2)$. Notice that when $\alpha=1$, we obtain the direct product of $\mathbb{Z}_p$ and $\mathbb{Z}_q$.

We now describe its irreps. There are $q$ one dimensional irreps and $(p-1)/q$, $q$ dimensional irreps. The one dimensional irreps are obtained as the extension of the trivial irrep of $\mathbb{Z}_p$ to $G$ and tensored with each of the $q$ irreps of $\mathbb{Z}_q$. The $q$ higher dimensional irreps are obtained as induced representations from a non-trivial irrep of $\mathbb{Z}_p$. The irreps of $\mathbb{Z}_p$ are characterized by $k$ and are of the form $\exp(2\pi ika/p)=\omega_p^{ka}$ for a group element $a$. All those $k$ in the orbit of $k\alpha^b$ for $b\in \mathbb{Z}_q$ will be induced to the same irrep of $G$. The induction can be given as
\[
\rho_k(a,b) = \sum_{s\in \mathbb{Z}_q} \omega_p^{ka\alpha^{-s}} \ket{s}\bra{s-b}\,.
\]

We now describe the centralizers of $G$. The centralizer of the identity element $(0,0)$ is $G$. The centralizer of any element $(a,b)$, where $a\neq 0$ is $\mathbb{Z}_p$. The centralizer of $(0,b)$, where $b\neq 0$ is $\mathbb{Z}_q$.

The Fourier transform is now easy to construct. Notice that since two of the centralizers are abelian, their QFT is efficient. For the QFT over $G$, first take the group basis $\ket{a,b}$ and convert it into the basis $\ket{z,t}$, where $z\in \mathbb{Z}_p$ and $t$ is a transversal which can be picked to be an element of $\mathbb{Z}_q$. Now perform a QFT over $\mathbb{Z}_p$ in the first register. Then conditioned on the value in the first register perform a second transform as follows. If the value in the first register is a non-trivial irrep of $\mathbb{Z}_p$, then do nothing since we already have an irrep of $G$. If the first register has a trivial irrep of $\mathbb{Z}_p$, then perform a QFT over $\mathbb{Z}_q$ in the second register. This gives us a QFT over $G$. Since we can now perform efficient QFTs over all the centralizers, we can perform efficient QFT over $\D(G)$.

For the Clebsch-Gordan transform, as the first condition, we need to perform CG transforms over all centralizers and their intersections. Since all intersections and two of the centralizers are abelian, their CG transforms are efficient. For the CG transform over $G$, we only have to consider the case of high dimensional irreps. For these irreps, the tensor product looks like
\[
(\rho_k\otimes\rho_l) (a,b) = \sum_{s,t} \omega_p^{a(k\alpha^{-s}+l\alpha^{-t})}\ket{s,t}\bra{s-b,t-b}\,.
\]
The Clebsch-Gordan transform, if $k+l\neq 0$, is
\[
\ket{s,t}\longrightarrow\ket{t-s,t}\,.
\]
This produces the state
\[
\sum_{s,t} \omega_p^{a(k\alpha^{-t+s}+l\alpha^{-t})}\ket{s,t}\bra{s,t-b} = \sum_{s,t} \omega_p^{a\alpha^{-t}(k\alpha^{s}+l)}\ket{s,t}\bra{s,t-b}\,,
\]
which is a direct sum of $q$ irreps $\rho_{k\alpha^{s}+l}$ for $s\in\mathbb{Z}_q$. Now, if $k+l=0$, then again we first perform the above transform to get
\[
\sum_{s,t} \omega_p^{a\alpha^{-t}(k\alpha^{s}-k)}\ket{s,t}\bra{s,t-b}\,.
\]
Notice that when $s=0$, this is the regular representation of $\mathbb{Z}_q$. Therefore, we have to perform a conditional QFT (conditioned on $s=0$) additionally. Then we would obtain a direct sum of the $q$ one dimensional irreps and $q-1$ high dimensional irreps $\rho_{k(\alpha^{s}-1)}$ for $s\in \mathbb{Z}_q$, $s\neq 0$. This completes the CG transform over $G$. We can now perform QFT and CG transforms over centralizers and intersections of centralizers.

To complete the conditions for CG transform over $\D(G)$, we need to block diagonalize $\rho_k$ restricted to intersections of centralizers i.e., block diagonalize $\rho_k$ when restricted to $\mathbb{Z}_p$ and $\mathbb{Z}_q$. From the structure of $\rho_k$, it is easy to see that when restricted to $\mathbb{Z}_p$, it is already diagonal and when restricted to $\mathbb{Z}_q$, it is the regular representation of $\mathbb{Z}_q$. Therefore, a QFT over $\mathbb{Z}_q$ would diagonalize it. Thus we can perform efficient QFT and CG transforms over $\D(G)$.

\section{Conclusions}\label{conclusions}
In this paper, we gave an efficient circuit for the quantum Fourier transform over $\D(G)$, the quantum double of a finite group $G$ and show how to apply it to $\D(S_n)$. We used this circuit to give efficient algorithms for approximating link invariants. We then showed some hardness results for approximating and exact evaluation of link invariants arising from $\D(G)$. We showed that additive approximations of link invariants arising from irreps of $\D(G)$ are $\BPP$-hard and multiplicative approximations are $\SBP$-hard and exact evaluations are $\sP$-hard. We also gave an efficient randomized algorithm to additively approximate the link invariants when the conjugacy class (or fluxon) irrep is used. This shows that for the fluxon irrep the problem is $\BPP$-complete. We then gave an efficient circuit for the Clebsch-Gordan transform for fluxon irreps of $\D(G)$ and, under certain conditions, for general irreps. We gave an example of a quantum group, namely $\D(\mathbb{Z}_p\rtimes\mathbb{Z}_q)$ (which is powerful enough to do universal quantum computation) for which we show how to perform the QFT and CG transforms. We also showed how to simulate topological quantum computation inside exponentially large irreps of $\D(G)$ efficiently using the Fourier and Clebsch-Gordan transforms. 

In our hardness results, we need that the size of $G$ be constant. An immediate question is how to extend this to asymptotically growing group sizes. One way is to make the procedure to kill unwanted solutions in our proof more efficient. Another question is - to which groups can these hardness results be extended. It is possible that the hardness results proved here for $A_n$ are also true for certain non-solvable groups which are not simple. For example, it is immediately true for $S_n$ by restricting to conjugacy classes which are also in $A_n$. But it could be true for a more general class of non-solvable groups. It would also be interesting to develop the Clebsch-Gordan transform over other groups $G$ for arbitrary irreps of $\D(G)$. This would be useful from the point of view of simulation of anyons and for the development of quantum circuits. Finally, it would be interesting to determine the power of a $\D(G)$ computer where we do not use post-selection (since with post-selection, it is universal for quantum computation for certain groups). Another interesting  question regarding post selection is, what class of operations or gates can be extended to universal quantum computation when post selection is used.

\section*{Acknowledgements} 
We thank Greg Kuperberg for useful discussions, especially for telling us about the class $\SBP$ and for suggesting that multiplicative approximations to $\D(G)$ invariants might be $\SBP$ complete. We also thank one of the anonymous referees of QIP for valuable suggestions. We thank Gorjan Alagic for telling us about \cite{RW} and \cite{FRW}. We acknowledge the support of NSF grants 1117427 and 0829917 and ARO contract W911NF-04-R-0009 at the University of Connecticut.

\bibliography{qd}

\end{document}